\pdfoutput=1
\documentclass[final,finalrunningheads,a4paper,orivec]{llncs}
\usepackage{xpatch}

\newtoggle{anonymous}
\togglefalse{anonymous}

\newtoggle{arxiv}
\toggletrue{arxiv}

\usepackage[utf8]{inputenc}
\DeclareMathSymbol{:}{\mathpunct}{operators}{"3A}

\usepackage{breakurl}
\usepackage{eucal}

\usepackage{centernot}
\usepackage{mdframed}

\usepackage{amsfonts} 
\usepackage{amssymb}

\usepackage{pifont}% http://ctan.org/pkg/pifont
\usepackage{multirow}

\usepackage{paralist}
\usepackage{mathtools}
\usepackage{url}
\usepackage{extarrows}

\usepackage{etoolbox}
\usepackage{xspace} % Smart spaces after commands

\usepackage{stmaryrd} % for llbracket and rrbracket\textbf{\textbf{}}

\usepackage[inline,shortlabels]{enumitem} 
\newlist{inlinelist}{enumerate*}{1}
\setlist*[inlinelist,1]{%
  label=(\roman*),
}

\usepackage{xifthen}  % Extended conditional commands
\usepackage{ifthen}

\usepackage{nicefrac}
\usepackage{caption,subcaption}

\usepackage{color}

\usepackage{tikz} % for HSL
\usepackage{tikz-cd} % for commutative diagrams
\usepackage{float}

\usepackage{chngcntr}
\usepackage{apptools}
\AtAppendix{\counterwithin{definition}{section}}
\AtAppendix{\counterwithin{example}{section}}
\AtAppendix{\counterwithin{lemma}{section}}
\AtAppendix{\counterwithin{figure}{section}}

\usepackage{hyperref}
\usepackage{cleveref}

\usepackage[final,nomargin,inline,index]{fixme} % Simplified management of FIXME's
\fxusetheme{color}

\definecolor{Black}{HTML}{000000}
\definecolor{Gray}{HTML}{808080}
\definecolor{Magenta}{HTML}{FF00FF}
\definecolor{RubineRed}{HTML}{ED017D}
\definecolor{ForestGreen}{HTML}{028A0F}
\definecolor{OliveGreen}{HTML}{808000}
\definecolor{MidnightBlue}{HTML}{006795}
\definecolor{Plum}{HTML}{92268F}

\FXRegisterAuthor{bart}{anbart}{\color{magenta} {\underline{bart}}}
\FXRegisterAuthor{em}{anem}{\color{violet} {\underline{emily}}}

\hypersetup{
  breaklinks   = true,
  colorlinks   = true, %Colours links instead of ugly boxes
  urlcolor     = blue, %Colour for external hyperlinks
  linkcolor    = blue, %Colour of internal links
  citecolor    = red   %Colour of citations
}
\hypersetup{final} % Needed to generate hyperlinks even in draft mode (ERROR IN MAC)

\usepackage{listings}
\definecolor{listingBG}{HTML}{FFFFCB}%
\definecolor{listingFrame}{HTML}{BBBB98}%
\definecolor{listingLineno}{rgb}{0.5,0.5,1.0}%
\crefname{lstlisting}{Listing}{Listings}  

\definecolor{LightGrey}{rgb}{0.985,0.985,0.985}
\definecolor{LstBackground}{rgb}{0.975,0.975,0.975}

\lstdefinelanguage{txscript}{
	commentstyle=\color{Gray},
	morecomment=[l]{//},
	morecomment=[s]{/*}{*/},
	classoffset=0,
        escapeinside={(*}{*)},
	morekeywords={if,then,else,for,in,transfer,pay,pays,contract,skip,require,mapping,int,uint,address,abort,return},
	keywordstyle=\color{Plum},
	classoffset=1,
	morekeywords={sender},
	keywordstyle=\addrColor,
	classoffset=2,        
	morekeywords={origin},
	keywordstyle=\pmvColor,
	classoffset=3,        
        morekeywords={constructor},
	keywordstyle=\color{MidnightBlue}\bfseries,
	basicstyle=\fontseries{m}\normalsize\ttfamily
	\lst@ifdisplaystyle\scriptsize\fi,
}

\newcommand{\ifempty}[3]{%
  \ifthenelse{\isempty{#1}}{#2}{#3}%
}

\newcommand{\ifdots}[3]{%
  \ifthenelse{\equal{#1}{...}}{#2}{#3}%
}

\newtoggle{hidden}
\toggletrue{hidden}

\makeatletter
\newcommand*{\itemequation}[3][]{%
  \item
  \begingroup
    \refstepcounter{equation}%
    \ifx\\#1\\%
    \else  
      \label{#1}%
    \fi
    \sbox0{#2}%
    \sbox2{$\displaystyle#3\m@th$}%
    \sbox4{\@eqnnum}%
    \dimen@=.5\dimexpr\linewidth-\wd2\relax
    \ifcase
        \ifdim\wd0>\dimen@
          \z@
        \else
          \ifdim\wd4>\dimen@
            \z@
          \else 
            \@ne
          \fi 
        \fi
      \@latex@warning{Equation is too large}%
    \fi
    \noindent   
    \rlap{\copy0}%
    \rlap{\hbox to \linewidth{\hfill\copy2\hfill}}%
    \hbox to \linewidth{\hfill\copy4}%
    \hspace{0pt}% allow linebreak
  \endgroup
  \ignorespaces 
}
\makeatother

\newcommand{\Real}[1]{\mathrm{Real}}

\newcommand{\codefont}{\fontsize{9}{9}\selectfont}
\newcommand{\code}[1]{{\tt\codefont{#1}}}
\newcommand{\contract}[2][]{{\tt\codefont{\cmvColor{#2_{#1}}}}}
\newcommand{\txcode}[1]{{\tt\codefont{\txColor{#1}}}}

\newcommand{\sender}{{\addrColor{\code{sender}}}}
\newcommand{\origin}{{\pmvColor{\code{origin}}}}
\newcommand{\callee}[1]{{\cmvColor{\it callee}({#1})}}
\newcommand{\deps}[1]{{\cmvColor{\it deps}}({#1})}

\newcommand{\Eg}{E.g.\@\xspace}
\newcommand{\eg}{e.g.\@\xspace}
\newcommand{\ie}{i.e.\@\xspace}
\newcommand{\wrt}{w.r.t.\@\xspace}

\newcommand{\Wlog}{W.l.o.g.\@\xspace}

\renewcommand{\epsilon}{\varepsilon}

\newcommand{\fst}{\mathit{fst}}

\newenvironment{proofof}[2][]{%
  \ifempty{#1}
  {\subsection*{Proof of~\Cref{#2}}}
  {\subsection*{Proof of~\Cref{#2} ({#1})}}
  \label{#2-proof}
  }%
  {}

\def\addrColor{\color{magenta}}
\newcommand{\addrFmt}[1]{{\addrColor{\tt #1}}}
\newcommand{\addr}[2][]{\addrFmt{#2}_{\addrColor{#1}}\xspace}
\newcommand{\AddrA}[1][]{{\addrColor{\mathcal{A}_{#1}}}} % set of addresses
\newcommand{\AddrB}[1][]{{\addrColor{\mathcal{B}_{#1}}}} % set of addresses

\newcommand{\AddrU}[1][]{\addrFmt{\mathbb{A}}_{\addrColor{#1}}} % universe of addresses

\newcommand{\cmvOfcst}[1]{\dag{#1}}
\newcommand{\strip}[2]{{#1}\!\upharpoonright_{#2}}
\renewcommand{\strip}[2]{{#1} \cap \deps{#2}}

\def\pmvColor{\color{red}}
\newcommand{\pmvFmt}[1]{{\pmvColor{\tt #1}}}
\newcommand{\pmv}[2][]{\pmvFmt{#2}_{\pmvColor{#1}}\xspace}
\newcommand{\pmvA}[1][]{\pmv[{#1}]{A}} % user constant
\newcommand{\pmvB}[1][]{\pmv[{#1}]{B}}

\newcommand{\pmvM}[1][]{\pmv[#1]{M}} % adversary

\newcommand{\PmvM}[1][]{\pmvFmt{\mathcal{M}}_{\pmvColor{#1}}} % set of users

\newcommand{\Adv}{\PmvM} % MEV adversary
\newcommand{\PmvU}{\pmvFmt{\mathbb{A}}_{\pmvColor{u}}} % universe of users

\def\cmvColor{\color{blue}}
\newcommand{\cmvFmt}[1]{{\cmvColor{\code{#1}}}}
\newcommand{\cmv}[2][]{\cmvFmt{#2}_{\cmvColor{#1}}}
\newcommand{\cmvC}[1][]{\cmvFmt{C_{#1}}} % contract
\newcommand{\cmvCi}[1][]{\cmvFmt{C'}_{#1}}

\newcommand{\cmvD}{\cmv{D}} % contract

\newcommand{\CmvC}[1][]{\cmv[#1]{\mathcal{C}}} % set of contracts
\newcommand{\CmvCi}[1][]{\cmvFmt{\mathcal{C}'_{\cmvColor{{\rm #1}}}}}

\newcommand{\CmvD}[1][]{\cmv[#1]{\mathcal{D}}} % set of contracts
\newcommand{\CmvDi}[1][]{\cmvFmt{\mathcal{D}'_{#1}}} % set of contracts

\newcommand{\CmvU}[1][]{\cmvFmt{\mathbb{A}}_{\cmvColor{c}}} % universe of contracts

\def\cstColor{\color{MidnightBlue}}
\newcommand{\cstFmt}[1]{{\cstColor{#1}}}
\newcommand{\cst}[2][]{\cstFmt{#2}_{\cstColor{#1}}}
\newcommand{\cstC}[1][]{\cst[#1]{\Gamma}} % contract
\newcommand{\cstCi}[1][]{\cstFmt{\Gamma'_{\cstColor{{\rm #1}}}}}

\newcommand{\cstD}[1][]{\cst[#1]{\Delta}} % contract
\newcommand{\cstDi}[1][]{\cstFmt{\Delta'_{\cstColor{{\rm #1}}}}}

\def\txColor{\color{MidnightBlue}}

\newcommand{\txFmt}[1]{{\txColor{\sf #1}}}

\newcommand{\tx}[2][]{\txFmt{#2}_{\txColor{#1}}}
\newcommand{\txT}[1][]{\tx[#1]{X}} % transaction
\newcommand{\txY}[1][]{\tx[#1]{Y}} % transaction
\newcommand{\txTi}[1][]{\txFmt{X'_{\txColor{{\it #1}}}}}

\newcommand{\TxTS}[1][]{\vec{\txT[#1]}} % sequence of transactions
\newcommand{\TxYS}[1][]{\vec{\txY[#1]}} % sequence of transactions
\DeclareMathAlphabet{\mathbfsf}{\encodingdefault}{\sfdefault}{bx}{n}

\newcommand{\dom}[1]{\operatorname{dom} {#1}}

\newcommand{\Nat}{\mathbb{N}}

\newcommand{\setenum}[1]{\{#1\}}
\newcommand{\setcomp}[2]{\left\{{#1} \,\middle|\, {#2}\right\}}
\newcommand{\wmvA}[1][]{w_{#1}}

\newcommand{\WmvA}[1][]{W_{#1}}
\newcommand{\WmvAi}[1][]{W'_{#1}}

\newcommand{\waltok}[2]{#1:#2}
\newcommand{\walenum}[1]{[#1]}
\newcommand{\walpmv}[2]{{#1}\walenum{#2}}
\newcommand{\walu}[3]{\walpmv{#1}{\waltok{#2}{#3}}}

\newcommand{\waldistrarrow}[1]{\approx_{\$}}

\newcommand{\wealth}[2]{\$_{#1}\ifempty{#2}{}{({#2})}}
\newcommand{\idxfun}[1]{\mathbf{1}_{#1}}
\newcommand{\price}[1]{\$\idxfun{#1}}

\newcommand{\gain}[3]{\mathit{\gamma}_{#1}\ifempty{#2}{}{({#2},{#3})}}
\newcommand{\mall}[2]{{\kappa_{#1}\ifempty{#2}{}{({#2})}}}
\newcommand{\mev}[3]{{\mathrm{MEV\!}_{#1}\ifempty{#1#2#3}{}{({#2\ifempty{#3}{}{,#3}})}}}
\newcommand{\lmev}[3]{{\mathrm{MEV\!}_{#1}\ifempty{#1#2#3}{}{({#2\ifempty{#3}{}{,#3}})}}}

\newcommand{\qnonint}[2]{\ifempty{#1}{\mathcal{I}}{\mathcal{I}({#1} \rightsquigarrow {#2})}}
\newcommand{\intok}[2]{{\tokColor{\mathit{in}}}_{{#1}}({#2})}
\newcommand{\outtok}[2]{{\tokColor{\mathit{out}}}_{{#1}}({#2})}

\newcommand{\qedex}{\ensuremath{\diamond}}

\definecolor{LightGrey}{rgb}{0.95,0.95,0.95}
\definecolor{keyword}{HTML}{7F0055}

\def\tokColor{\color{ForestGreen}}
\newcommand{\tokFmt}[1]{{\tokColor{\tt #1}}}
\newcommand{\ETH}{\tokFmt{ETH}}
\newcommand{\tok}[2][]{\tokFmt{#2}_{\tokColor{#1}}\xspace}

\newcommand{\tokT}[1][]{\tok[{#1}]{T}}    % token constant
\newcommand{\tokTi}[1][]{\tok[{#1}]{T'}}
\newcommand{\TokU}{\tokFmt{\mathbb{T}}} % universe of all tokens

\newcommand{\tokcT}[1][]{\contract{cT}}  
\newcommand{\tokbT}[1][]{\contract{bT}}  
\newlength\replength
\newcommand\repfrac{.1}

\newcommand\rulewidth{.6pt}
\newcommand\tdashfill[1][\repfrac]{\cleaders\hbox to \replength{%
  \smash{\rule[\arraystretch\ht\strutbox]{\repfrac\replength}{\rulewidth}}}\hfill}

\newcommand\tdotfill[1][\repfrac]{\cleaders\hbox to \replength{%
  \smash{\raisebox{\arraystretch\dimexpr\ht\strutbox-.1ex\relax}{.}}}\hfill}

\def\sysColor{\color{Black}}

\newcommand{\sysFmt}[1]{{\sysColor{#1}}}
\newcommand{\sysS}[1][]{\mathord{\sysFmt{S}_{\sysColor{#1}}}}

\newcommand{\sysSi}[1][]{\mathord{\sysColor{\sysS'_{#1}}}}

\makeatletter
\newcommand{\mtmathitem}{%
\xpatchcmd{\item}{\@inmatherr\item}{\relax\ifmmode$\fi}{}{\errmessage{Patching of \noexpand\item failed}}
\xapptocmd{\@item}{$}{}{\errmessage{appending to \noexpand\@item failed}}}
\makeatother

\title{A quantitative notion of economic security for smart contract compositions}

\iftoggle{anonymous}{}
{\author{Emily Priyadarshini\inst{1} \and Massimo Bartoletti\inst{2}}
\institute{IISER Pune, India \and University of Cagliari, Italy}
}

\begin{document}

\maketitle

\begin{abstract}
Decentralized applications are often composed of multiple interconnected smart contracts.
This is especially evident in DeFi, where protocols are heavily intertwined and rely on a variety of basic building blocks such as tokens, decentralized exchanges and lending protocols.
A crucial security challenge in this setting arises when adversaries target individual components to cause systemic economic losses.
Existing security notions focus on determining the existence of these attacks, but fail to quantify the effect of manipulating individual components on the overall economic security of the system.
In this paper, we introduce a quantitative security notion that measures how an attack on a single component can amplify economic losses of the overall system.
We study the fundamental properties of this notion and apply it to assess the security of key compositions.
In particular, we analyse under-collateralized loan attacks in systems made of lending protocols and decentralized exchanges.
\end{abstract}
\section{Introduction}
\label{sec:intro}

Developing decentralized applications nowadays involves suitably designing, assembling and customizing a multitude of smart contracts, resulting in complex interactions and dependencies.
In particular, recent DeFi applications are highly interconnected compositions of smart contracts of various kinds, including tokens, derivatives, decentralized exchanges (DEX), and lending protocols~\cite{Kitzler22fc,Kitzler23tweb}.

This complexity poses significant security risks, as adversaries targeting one of the components may compromise the security of the overall application.
Note that, for this to happen, the attacked component does not even need to have a proper vulnerability to exploit. 
For example, in an application composed of a lending protocol and a DEX serving as a price oracle, adversaries could target the DEX in order to artificially inflate the price of an asset that they have previously deposited to the lending pool. This manipulation would allow adversaries to borrow other assets with an insufficient collateral, circumventing the intended economic mechanism of the lending protocol~\cite{Gudgeon2020cvcbt,Qin21fc,BCL21wtsc,Mackinga22icbc,Arora24asiaccs}.

The first step to address these risks is to formally define when a system of smart contracts is secure.
In recent years, a few security notions have emerged, starting from Babel, Daian, Kelkar and Juels' ``Clockwork finance''~\cite{Babel23clockwork}.
Broadly, these definitions try to characterise the economic security of smart contract systems based on the extent of economic damage that adversaries can inflict on them.
In this context, adversaries are typically assumed to have the powers of consensus nodes.
Namely, they can reorder, drop or insert transactions in blocks.
Accordingly, the economic damage on a system $\sysS$ can be quantified in terms of the Maximal Extractable Value (MEV) that adversaries can extract from $\sysS$ by leveraging these powers~\cite{Daian20flash}.
To provide a more concrete formulation of the existing  notions, consider a set of contracts $\cstD$ to be deployed in a system $\sysS$.
We denote by \mbox{$\sysS \mid \cstD$} the system composed of $\sysS$ and~$\cstD$.
The security criterion in~\cite{Babel23clockwork} requires that
$\mev{}{\sysS \mid \cstD}{} \leq (1+\epsilon) \, \mev{}{\sysS}{}$: namely,
the MEV extractable from \mbox{$\sysS \mid \cstD$} does not exceed the MEV extractable from $\sysS$ by more than a factor of~$\epsilon$.
This notion does not capture our intuition of assessing the security of $\cstD$ in terms of the economic losses that $\cstD$ could incur due to adversaries interacting with the context $\sysS$.  % 
For example, an airdrop contract $\cstD$ that gives away tokens would be deemed insecure, while in reality its interactions with $\sysS$ are irrelevant.

In a different security setting, a similar intuition was the basis of Goguen and Meseguer' non-interference~\cite{GoguenMeseguer82sp}, which was originally formulated as follows: 
\begin{quote}
``One group of users, using a certain set of commands, is noninterfering with another group of users if what the first group does with those
commands has no effect on what the second group of users can see''.
\end{quote}

In the setting of smart contract compositions, this notion can be reinterpreted by requiring that adversaries interacting with $\sysS$ do not inflict economic damage to $\cstD$.
The notion of \emph{MEV non-interference} introduced by~\cite{BMZ24fc} is based on this idea, using MEV as a measure of economic damage. 
The approaches in~\cite{Guesmi24dlt,Yao24scif} are also based on the idea of non-interference, but replacing MEV with an explicit tagging of contract variables as high-level or low-level variables.

A common aspect of these approaches to economic non-interference is their \emph{qualitative} nature: namely, these definitions classify a composition as either secure or insecure, in a binary fashion.  
While a qualitative evaluation is sufficient when a composition is deemed secure, in that case that it is not, 
% it fails to give 
it does not provide any meaningful estimate of the \emph{degree} of interference.
For example, in the insecure composition between a lending protocol and a DEX mentioned above, a quantitative measure could provide insights into the extent to which the system state (\eg, the liquidity reserves in the DEX) and the contract parameters (\eg, the collateralization threshold) contribute to increasing the economic loss. 

\paragraph{Contributions}

This paper introduces a quantitative notion of economic security for smart contract compositions. 
Our \emph{MEV interference}, which we denote by $\qnonint{\sysS}{\cstD}$, measures the increase of economic loss of contracts $\cstD$ that adversaries can achieve by manipulating the context $\sysS$.
We apply our notion to assess the security of some notable contract compositions, including a bet on a token price, and a lending protocol relying on a DEX as a price oracle.
We prove some fundamental properties of our notion:
more specifically, $\qnonint{\sysS}{\cstD}$ increases when $\sysS$ is extended with contracts that are not in the dependencies of $\cstD$ (\Cref{prop:qnonint:larger-state});
$\qnonint{\sysS}{\cstD}$ does not depend on the token balances of users except adversaries (\Cref{prop:qnonint:adv wallets});
$\qnonint{\sysS}{\cstD}$ is preserved when extending $\sysS$ with contracts $\cstC$ that enjoy some specific independency conditions with respect to $\cstD$ (\Cref{th:qnonint:preserving-interference}).
% We finally discuss the limitations of our approach.
\section{Smart contracts model}
\label{sec:model}

\begin{table}[t]
\caption{Summary of notation.}
\label{tab:notation}
\centering
\begin{tabular}{p{40pt}p{100pt}p{60pt}p{130pt}}
\hline
$\pmvA,\pmvB$ & User accounts 
& $\AddrA,\AddrB$ & Sets of [user$\mid$contract] accounts
\\
$\cmvC,\cmvD$ & Contract accounts 
& $\CmvC,\CmvD$ & Sets of contract accounts
\\
$\tokT,\tokTi$ & Token types 
& $\price{\tokT}$ & Price of $\tokT$
\\
$\txT,\txTi$ & Transaction names 
& $\pmvA:\contract{C}.\txcode{f}(\code{args})$ & Transaction
\\
$\sysS,\sysSi$ & Blockchain states 
& $\wealth{\CmvC}{\sysS}$ & Wealth of contracts $\CmvC$ in $\sysS$
\\
$\WmvA,\WmvAi$ & Wallet states 
& $\deps{\CmvC}$ & Dependencies of contracts $\CmvC$
\\
$\cstC,\cstD$ & Contract states 
& $\cmvOfcst{\cstC}$ & Contract accounts in $\cstC$
\\
\hline
\end{tabular}
\end{table}

We consider a contract model inspired by account-based platforms such as Ethereum.
The basic building blocks of our model are
a set $\TokU$ of \emph{token types} ($\tokT, \tokTi, \ldots$),
representing crypto-assets (\eg, ETH), 
and a set $\AddrU$ of \emph{accounts}.
We partition accounts into
\emph{user accounts} $\pmvA, \pmvB, \ldots \in \PmvU$
(representing the so-called \emph{externally owned accounts} in Ethereum)
and
\emph{contract accounts} $\cmvC, \cmvD, \ldots \in \CmvU$.

The state of a user account is a map \mbox{$\wmvA \in \TokU \rightarrow \Nat$}
from token types to non-negative integers, representing a \emph{wallet} of tokens.
The state of a contract account is a pair $(\wmvA,\sigma)$,
where $\wmvA$ is a wallet and $\sigma$ is a key-value map, representing the contract storage.
A \emph{blockchain state} $\sysS$
is a map from accounts to their states. 
% where user wallets include at least $\Adv$'s.
We write an account state in square brackets, wherein we denote by $\waltok{n}{\tokT}$ a balance of $n$ units of token $\tokT$ in the wallet, and by $\code{x}=v$, the association of value $v$ to the storage variable \code{x}.
For example, 
$\walpmv{\contract{C}}{\waltok{1}{\tokT},\code{owner}=\pmvA}$
represents a state where the contract $\contract{C}$ stores 1 unit of $\tokT$, and the variable $\code{owner}$ contains the address $\pmvA$.
We write a blockchain state as the composition of its account states, using the symbol $\mid$ as a separator. 
For example,
\(
\sysS =
\walpmv{\pmvA}{\waltok{1}{\tokT},\waltok{2}{\ETH}} \mid
\walpmv{\contract{C}}{\waltok{1}{\tokT},\code{owner}=\pmvA}
\)
is a state composed by a user account and a contract account.

Contracts are made up of a finite set of \emph{functions}, which can be called by \emph{transactions} sent by users.
A function can:
\begin{inlinelist}
\item receive parameters and tokens from the caller,
\item transfer tokens to user accounts (including the caller),
\item update the contract state,  
\item call other functions (possibly of other contracts, and possibly transferring tokens along with the call),
\item return values to the caller.
\end{inlinelist}
Functions can only manipulate tokens as described above: in particular, they cannot mint or burn tokens, or drain tokens from other accounts.
Transactions $\txT, \txTi, \ldots$ are calls to contract functions,
written $\pmvA:\contract{C}.\txcode{f}(\code{args})$,
where $\pmvA$ is the user signing the transaction,
$\contract{C}$ is the called contract,
$\txcode{f}$ is the called function,
and $\code{args}$ is the list of actual parameters.
Parameters can also include transfers of tokens $\tokT$ from $\pmvA$ to $\cmvC$,
written $\pmvA\ \code{pays}\ \waltok{n}{\tokT}$.
Invalid transactions are reverted (\ie, they do not update the blockchain state).
We remark that our security definition and results do not rely on a particular language for functions:  we just assume a deterministic transition relation $\xrightarrow{}$ between blockchain states, where state transitions are triggered by transactions.
To write examples, however, we will instantiate this abstract model using a contract language inspired by Solidity.

We assume that a contract $\cmvD$ can call a function of a contract $\cmvC$ only if $\cmvC$ was deployed before $\cmvD$.
Formally, defining $\cmvC \prec \cmvD$ (read: ``\emph{$\cmvC$ is called by $\cmvD$}'') when some function in $\cmvD$
calls some function in $\cmvC$, we require that the transitive and reflexive closure $\sqsubseteq$ of $\prec$ is a partial order.
We define the \emph{dependencies} of a contract $\cmvC$ as $\deps{\cmvC} = \setcomp{\cmvCi}{\cmvCi \sqsubseteq \cmvC}$, and extend this notion to \emph{sets} of contracts $\CmvC$.
% $\CmvC = \setenum{\cmvC[1],\ldots,\cmvC[n]}$. 
We assume that blockchain states $\sysS$ enjoy the following conditions:
\begin{inlinelist}
\item $\sysS$ contains all its dependencies, \ie if $\cmvC$ is a contract in $\sysS$, then also the contracts $\deps{\cmvC}$ are in~$\sysS$;
\item $\sysS$ contains \emph{finite tokens}. 
% \ie $\sum_{\pmvA,\tokT} \sysS(\pmvA) (\tokT) \in \Nat$.
\end{inlinelist}
All states mentioned in our results are assumed to enjoy these well-formedness assumption.%
\footnote{Note that well-formedness rules out some problematic features like reentrancy, which instead is present in Ethereum. However, reentrancy can always be removed by using suitable programming patterns, so we do not consider this as a limitation.}
We write \mbox{$\sysS = \WmvA \mid \cstC$}
for a blockchain state $\sysS$ composed 
of user wallets $\WmvA$ and contract states $\cstC$.
We can deconstruct wallets, writing \mbox{$\sysS = \WmvA \mid \WmvAi \mid \cstC$} when the accounts in $\WmvA$ and $\WmvAi$ are disjoint, as well as contract states, writing \mbox{$\sysS = \WmvA \mid \cstC \mid \cstD$}.
We denote by $\cmvOfcst{\cstC}$ the set of contract accounts
in $\cstC$, 
\ie $\cmvOfcst{\cstC} = \dom{\cstC}$.
% and let $\deps{\cstD} = \deps{\cmvOfcst{\cstD}}$.
For example,
\(
\cmvOfcst{(\walpmv{\contract{C}}{\cdots} \mid \walpmv{\contract{D}}{\cdots})} = \setenum{\contract{C},\contract{D}}
\).
Given $\txT = \pmvA:\contract{C}.\txcode{f}(\code{args})$,
we write $\callee{\txT}$ for the target contract~$\contract{C}$.
\section{Threat model}

To define economic security of smart contract compositions, following~\cite{Babel23clockwork} we consider the Maximal Extractable Value (MEV) that can be extracted when new contracts $\CmvC$ are deployed in a blockchain state \mbox{$\sysS = \WmvA \mid \cstC$}, leading to a new state \mbox{$\sysS \mid \cstC \mid \cstD$} where $\cstD$ contains the initial state of the new contracts~$\CmvC$. 
Since our goal is measuring the loss of the new contracts $\cstD$ caused by attacking their dependencies $\cstC$, rather than considering the overall MEV of $\sysS \mid \cstD$, we isolate the MEV extractable from $\cstD$ and compare it to the MEV that could be extracted from $\cstD$ \emph{without} exploiting the dependencies $\cstC$. 
To this purpose, we leverage the adversary model and the notion of \emph{local MEV} introduced in~\cite{BMZ24fc}.

We start by designating a finite subset  $\Adv$ of user accounts as adversaries.
We assume that adversaries have full control of the selection and ordering of transactions --- a standard assumption in definitions of MEV~\cite{Babel23clockwork}.
Then, to measure the economic loss of a set of contracts $\CmvC$, we consider the wealth of $\CmvC$ in a blockchain state before and after the attack.
The wealth of $\CmvC$ in $\sysS$, written $\wealth{\CmvC}{\sysS}$, is given by the amount of tokens in each contract $\cmvC \in \CmvC$ in $\sysS$ weighted by their prices.
Recalling that a contract state is a pair $(w,\sigma)$ whose first element is a wallet, and denoting by $\price{\tokT}$ the price of a token type $\tokT$, the wealth of a single contract state $\walpmv{\cmvC}{w,\sigma}$ is given by $\sum_{\tokT} \wmvA(\tokT) \cdot \price{\tokT}$, \ie the summation, for all token types $\tokT$, of the number of tokens $\tokT$ in the wallet of $\cmvC$, times the price of $\tokT$.%
\footnote{Here we implicitly assume that the prices of \emph{native} crypto-assets are constant, 
since they do not depend on the blockchain state. We discuss this assumption in~\Cref{sec:conclusions}.}
By extending this to the set $\CmvC$, we obtain the following general definition of wealth:
% The wealth of a set of contracts $\CmvC$ in a blockchain state $\sysS = \WmvA \mid \cstC$ is given by:
\begin{equation}
    \label{eq:wealth}
    % \textstyle
    \wealth{\CmvC}{\sysS}
    \; = \;
    \sum_{\cmvC \in \CmvC, \tokT}
    \hspace{-2pt}    
    \fst(\cstC(\cmvC))(\tokT) \cdot \price{\tokT}    
\end{equation}

Building on the definition of wealth, we now revisit the notion of local MEV introduced in~\cite{BMZ24fc}.
The local MEV extractable by a set of contracts $\CmvC$ in a blockchain state $\sysS$, denoted by $\lmev{}{\sysS}{\CmvC}$, is the maximum loss that adversaries can inflict to $\CmvC$ by performing an arbitrary sequence of transactions crafted using their knowledge. 
By denoting with $\mall{}{\Adv}$ the set of transactions craftable by $\Adv$, this amounts to the maximum loss
$\wealth{\CmvC}{\sysS} - \wealth{\CmvC}{\sysSi}$
over all possible states $\sysSi$ reachable through a sequence $\TxTS$ of transactions in $\mall{}{\Adv}$. In symbols:
\begin{equation}
  \label{eq:lmev:unrestricted}
    \lmev{}{\sysS}{\CmvC}
    = \max \setcomp
    {
    % -\gain{\CmvC}{\sysS}{\TxTS}
    \wealth{\CmvC}{\sysS} - \wealth{\CmvC}{\sysSi}
    }
    {
    \TxTS \in \mall{}{\Adv}^*, \ 
    \sysS \xrightarrow{\TxTS} \sysSi
    }
\end{equation}

In $\lmev{}{\sysS}{\CmvC}$, adversaries are allowed to call any contract in $\sysS$, including the dependencies of $\CmvC$ not defined in $\CmvC$ itself.
This follows from the fact that $\mall{}{\Adv}$ does not pose any restriction on the callee of the transactions craftable by $\Adv$.
To estimate the MEV extractable from $\cstD$ \emph{without} exploiting the dependencies $\cstC$, we introduce an additional parameter $\CmvD$ to local MEV, representing the set of contracts callable by~$\Adv$.
We denote by
\(
\mall{\CmvD}{\Adv}
= 
\setcomp{\txT \in \mall{}{\Adv}}{\callee{\txT} \in \CmvD}
\)
the set of transactions craftable by $\Adv$ 
and targeting contracts in $\CmvD$.
We define:
% the local MEV of $\CmvC$ restricted to $\CmvD$ in $\sysS$ as:
\begin{equation}
    \label{eq:lmev}
    \lmev{\CmvD}{\sysS}{\CmvC}
    = \max \setcomp
    {
    % -\gain{\CmvC}{\sysS}{\TxTS}
    \wealth{\CmvC}{\sysS} - \wealth{\CmvC}{\sysSi}
    }
    {
    \TxTS \in \mall{\CmvD}{\Adv}^*, \ 
    \sysS \xrightarrow{\TxTS} \sysSi
    }
\end{equation}

Note that by the finite token assumption in~\Cref{sec:model},
the wealth is always finite, and so also the local MEV.

% \begin{example}
% \bartnote{mini-example to illustrate def}
% \end{example}

\section{A quantitative notion of economic security}
\label{sec:nonint}

In this section we introduce our notion of quantitative 
security for smart contract compositions, and study its theoretical properties. 
In~\Cref{sec:use-cases} we will apply it to analyse some archetypal compositions and attacks. 
% Before doing that, we motivate our definition by discussing a few intuitions about its properties.

Let $\sysS$ be a blockchain state, formed by users' wallets $\WmvA$ and contract states~$\cstC$,
where we want to deploy new contracts with an initial state $\cstD$.
Note that, by the well-formedness assumption introduced in~\Cref{sec:model}, the dependencies of $\cstD$ must be included in $\cstC \mid \cstD$, \ie any function call made by a contract in $\cstD$ must target some contracts in $\cstC$ or in $\cstD$.
We want to measure the security of the composition $\sysS \mid \cstD$ by analysing the additional loss that an adversary can inflict to the contracts in $\cstD$ by manipulating the dependencies $\cstC$. 
To this purpose, our definition will compare: 
\begin{itemize}

\item $\lmev{}{\sysS\mid\cstD}{\cmvOfcst{\cstD}}$, the maximal loss of the contracts in $\cstD$, where adversaries are able to send transactions to \emph{any} contract in $\sysS \mid \cstD$;

\item $\lmev{\cmvOfcst{\cstD}}{\sysS\mid\cstD}{\cmvOfcst{\cstD}}$, the maximal loss of the contracts in $\cstD$, where adversaries can \emph{only} send transactions to contracts in $\cstD$. Note that interactions between $\cstD$ and $\cstC$ are still possible, as contracts in $\cstD$ can invoke functions of contracts in $\cstC$ (\emph{``contract dependencies''}), and adversaries can extract tokens from $\cstC$ to play them in calls to contracts in $\cstD$ (\emph{``token dependencies''}).

\end{itemize}

Our security notion, called \emph{MEV interference},
measures how leveraging the dependencies in $\sysS$ can amplify the loss caused to~$\cstD$. 
We denote with \mbox{$\qnonint{\sysS}{\cstD}$}
the MEV interference caused by a blockchain state $\sysS$ to~$\cstD$.

% Armed with our list of desiderata, we define MEV interference. 
% We will see in~\Cref{sec:qnonint:properties} that this notion aligns with all our intuitions. 

\begin{definition}[MEV interference]
\label{def:qnonint}
    For a blockchain state $\sysS$ and a contract state $\cstD$, we quantify the $\lmev{}{}{}$ interference caused by $\sysS$ on $\cstD$ as:
    \begin{align*}
        \qnonint{\sysS}{\cstD} 
        \; = \;
        \begin{cases}
        1 - \dfrac{\lmev{\cmvOfcst{\cstD}}{\sysS\mid\cstD}{\cmvOfcst{\cstD}}}{\lmev{}{\sysS\mid\cstD}{\cmvOfcst{\cstD}}}
        & \textit{if} \hspace{3mm} \lmev{}{\sysS\mid\cstD}{\cmvOfcst{\cstD}} \neq 0
        \\
        0 & \text{otherwise}
        \end{cases}
    \end{align*}
    % When $\lmev{}{\sysS\mid\cstD}{\cmvOfcst{\cstD}} = 0$, we define $\qnonint{\sysS}{\cstD} = 0$.
\end{definition}

Our notion is consistent with the notion of MEV non-interference in~\cite{BMZ24fc}, which classifies $\sysS$ and $\cstD$ as non-interferent if
$\lmev{\cmvOfcst{\cstD}}{\sysS\mid\cstD}{\cmvOfcst{\cstD}} = \lmev{}{\sysS\mid\cstD}{\cmvOfcst{\cstD}}$.
Namely, \mbox{$\qnonint{\sysS}{\cstD} = 0$} iff $\sysS$ and $\cstD$ are non-interferent according to~\cite{BMZ24fc}.

\begin{figure}[t]
\begin{lstlisting}[
,language=txscript
,morekeywords={Airdrop,fund,withdraw},classoffset=4
,morekeywords={a,b,A,Oracle},keywordstyle=\addrColor,classoffset=5
,morekeywords={t,t1,t2,T,tin,tout},keywordstyle=\tokColor,classoffset=6
,caption={A simple airdrop contract.}
,label={lst:airdrop}
]
contract (*$\contract{Airdrop}$*) {
  fund(pay x:T) { } // any user can deposit x:T to the contract
  withdraw(x) {     // any user can withdraw any amount x:T 
    require(balance(T)>=x); // check that the contract has at least x:T
    transfer(sender,x:T);   // transfer x:T to the caller
  }
}
\end{lstlisting}
\end{figure}

\begin{example}[Any/Airdrop]
Consider an instance $\cstD = \walpmv{\contract{Airdrop}}{\waltok{n}{\tokT}}$ of the airdrop contract in~\Cref{lst:airdrop}, to be deployed in an arbitrary blockchain state~$\sysS$.
Note that $\lmev{}{\sysS \mid \cstD}{\setenum{\contract{Airdrop}}} = n \cdot \price{\tokT}$, since the adversary can craft a transaction $\pmvM:\contract{Airdrop}.\txcode{withdraw}(n)$ to extract all the tokens from the contract.
The restricted $\lmev{\setenum{\contract{Airdrop}}}{\sysS \mid \cstD}{\setenum{\contract{Airdrop}}}$ is equal to the unrestricted one, 
since the adversary just needs to interact with $\contract{Airdrop}$. 
Therefore, if $n>0$: 
\[
\qnonint{\sysS}{\cstD} 
= 
1 - \frac{\lmev{\setenum{\contract{Airdrop}}}{\sysS \mid \cstD}{\setenum{\contract{Airdrop}}}}{\lmev{}{\sysS \mid \cstD}{\setenum{\contract{Airdrop}}}} 
= 
0
\]
The same holds if $n=0$.
This is consistent with our intuition, since the adversary does not need to exploit the dependencies in $\sysS$ to extract MEV from $\cstD$.
\hfill\qedex
\end{example}

\begin{figure}[t]
\begin{lstlisting}[
,language=txscript
,morekeywords={Airdrop,fund,withdraw,getFee,getOwner,setFee},classoffset=4
,morekeywords={a,b,A,Oracle},keywordstyle=\addrColor,classoffset=5
,morekeywords={t,t1,t2,T,tin,tout},keywordstyle=\tokColor,classoffset=6
,morekeywords={FeeManager},keywordstyle=\cmvColor,classoffset=7
,caption={A simple airdrop contract with fees.}
,label={lst:airdrop-fee}
]
contract (*$\contract{AirdropFee}$*) {
  fund(pay x:T) { } // any user can deposit x:T to the contract
  withdraw(x) {     // any user can withdraw any amount x:T (minus fee)
    require(balance(T)>=x);
    fee = floor((FeeManager.getFee() * x) / 100); // integer division
    transfer(sender, x-fee:T);              
    transfer(FeeManager.getOwner(), fee:T);
  }
}
contract FeeManager {
  constructor() { owner=sender; feeRate=1; }
  getOwner() { return owner; }
  getFee()   { return feeRate; }
  setFee(r)  { require (r>=0 && r<=100); feeRate=r; }
}
\end{lstlisting}
\end{figure}

\begin{example}[FeeManager/Airdrop]
Consider a variant of the airdrop contract, where each withdrawal requires the user to pay a proportional fee (\Cref{lst:airdrop-fee}).
To obtain the fee rate, the $\contract{AirdropFee}$ contract calls the $\contract{FeeManager}$ contract.
Assume that we want to deploy
$\cstD = \walpmv{\contract{AirdropFee}}{\waltok{n}{\tokT}}$
in a blockchain state $\sysS$ containing 
$\walpmv{\contract{FeeManager}}{\code{feeRate}=r}$.
The unrestricted $\lmev{}{}{}$ is $n \cdot \price{\tokT}$, 
since an adversary can set the fee to $0$ by calling 
$\contract{FeeManager}.\txcode{setFee}(0)$ and then withdraw the full balance of $\waltok{n}{\tokT}$ from $\contract{AirdropFee}$.
Instead, the restricted $\lmev{}{}{}$ only amounts to $(n - \lfloor \nicefrac{r \cdot n}{100} \rfloor) \cdot \price{\tokT}$,
since the adversary cannot call $\contract{FeeManager}$ to manipulate the fee rate.
Therefore, if $n>0$:
\[
\qnonint{\sysS}{\cstD} 
= 
1 - \frac{n - \lfloor \nicefrac{r \cdot n}{100} \rfloor}{n}
\leq
\frac{r}{100}
\]
This is coherent with our intuition: the closer the fee rate is to 100, the greater the difference between restricted and unrestricted MEV, and so the possibility for the attacker to inflict more damage to the contract.
\hfill\qedex
\end{example}

% \subsection{Properties of MEV interference}
% \label{sec:qnonint:properties}

We now study the theoretical properties of MEV interference.
Because of space constraints, we relegate the proofs of our statements to a technical report on ArXiV.
\Cref{lem:qnonint:basic} establishes a few basic properties of MEV interference:
its value is zero when the context $\sysS$ has no contracts and when $\cstD$ is empty; furthermore, the interference is always comprised between 0 and 1.

\begin{lemma} % [\textbf{Basic properties of $\qnonint{}{}$}] 
\label{lem:qnonint:basic}
% For all $\sysS,\cstD$:
\begin{inlinelist}
\item \label{lem:qnonint:basic:1}
$\qnonint{\sysS}{\emptyset} = 0$;

\item \label{lem:qnonint:basic:2}
$\qnonint{\WmvA \mid \emptyset}{\cstD} = 0$;

\item \label{lem:qnonint:basic:3}
$0 \leq \qnonint{\sysS}{\cstD} \leq 1$.
\end{inlinelist}
\end{lemma}

Note that $\qnonint{\sysS}{\cstD}$ ranges from a minumum $0$, representing the case where the context $\sysS$ is not useful to extract MEV from $\cstD$, 
to a maximum $1$, corresponding to the case where 
% \emph{all} the MEV extractable from $\cstD$ depends on making the adversary interact with the context.
the economic loss that can be inflicted to $\cstD$ is purely due to the interactions of the adversary with~$\sysS$.
Enclosing MEV interference into an interval is a design choice, which we illustrate with an example.
Let $\sysS$ be a state with an airdrop contract releasing $1:\tokT$, where we want to deploy a new contract $\cstD$ that, upon the payment of $\waltok{1}{\tokT}$, releases all its balance of $\waltok{n}{\ETH}$.
Assume that the adversary has no tokens $\tokT$, so that she needs to extract $\waltok{1}{\tokT}$ from the airdrop in order to extract MEV from $\cstD$.
If we measured the interference from $\sysS$ to $\cstD$ as the difference between unrestricted and restricted MEV, \ie:
\[
\qnonint{\sysS}{\cstD}
\; \stackrel{?}{=} \;
\lmev{}{\sysS\mid\cstD}{\cmvOfcst{\cstD}}
-
\lmev{\cmvOfcst{\cstD}}{\sysS\mid\cstD}{\cmvOfcst{\cstD}}
\]
then we would obtain that
$\qnonint{\sysS}{\cstD} = n \cdot \price{\ETH}$,
\ie the interference would be proportional to the $\ETH$ balance in $\cstD$.
We do not find this measure particularly insightful: 
after all, what we observe is just that \emph{all} the MEV extractable from $\cstD$ is due to the interaction with the context $\sysS$.
In general, under these conditions, our intuition is that the interference should take its maximum value. 

\medskip
\Cref{lem:qnonint:zero-wealth} 
% formalises~\Cref{intuition:zero-wealth-zero-qnonint}:
states that when the newly deployed contracts $\cstD$ have no wealth (\ie, when $\wealth{\cmvOfcst{\cstD}}{\cstD} = 0$), then they have no MEV interference with the context.

\begin{lemma}
\label{lem:qnonint:zero-wealth}
If $\wealth{\cmvOfcst{\cstD}}{\cstD} = 0$, then $\qnonint{\sysS}{\cstD} = 0$.
\end{lemma}

Of course, if $\cstD$ has zero wealth, no loss can be inflicted to $\cstD$, regardless of any potential  manipulation of its dependencies in $\sysS$.
This also underscores a fundamental aspect of our definition --- namely, that it measures what happens in \emph{specific} contract states, rather than in \emph{arbitrary} reachable states of a given contract. 
For this reason, our intuition is to have \mbox{$\qnonint{\sysS}{\cstD} = 0$} whenever $\cstD$ has zero wealth, while not ruling out the possibility of having $\qnonint{\sysSi}{\cstDi} > 0$ in a state $\cstDi$ where the contracts have been funded.

\medskip
\Cref{prop:qnonint:larger-state} says that widening a blockchain state $\sysS$ potentially increases MEV interference to newly deployed contracts $\cstD$.
Formally, this amounts to showing that $\qnonint{}{}$ is monotonic
\wrt the operation of adding contracts $\cstC$ to the context,
\ie \mbox{$\qnonint{\sysS}{\cstD} \leq \qnonint{\sysS \mid \cstC}{\cstD}$}.
Note that by the well-formedness assumption, the statement implicitly assumes that $\cstD$ has no dependencies in $\cstC$.

\begin{theorem}
\label{prop:qnonint:larger-state}
% When $\deps{\cstD} \cap \cmvOfcst{\cstC} = \emptyset$:
\(
\qnonint{\sysS}{\cstD} \leq \qnonint{\sysS \mid \cstC}{\cstD}
\)
\end{theorem}

% $\qnonint{\sysS}{\cstD}$ should not decrease
% if we extend $\sysS$ with contracts that are not dependencies of $\cstD$.

For illustration, consider a state $\sysS$ where we want to deploy new contracts $\cstD$, with an interference estimated as $\qnonint{\sysS}{\cstD}$.
Assume now that the deployment of $\cstD$ is front-run by that of another set of contracts $\cstC$.
Of course $\cstD$ cannot have dependencies in $\cstC$, since otherwise it would not be possible to deploy $\cstD$ in $\sysS$ (as this would violate the well-formedness assumption).
Now, the interference $\qnonint{\sysS \mid \cstC}{\cstD}$
could either be equal to $\qnonint{\sysS}{\cstD}$, or possibly increase when the adversary can drain tokens from $\cstC$ to inflict more loss to~$\cstD$.
\Cref{prop:qnonint:larger-state} states that, in any case, the interference should not decrease.

The following example shows a case where the inequality given by~\Cref{prop:qnonint:larger-state} is strict.
This is because, even if $\cstD$ has no \emph{contract} dependencies in $\cstC$, the adversary may exploit their \emph{token} dependencies, \ie extract tokens from $\cstC$ and leverage them to extract more tokens from~$\cstD$. 

\begin{example}
\label{ex:qnonint-larger state}
Let $\sysS = \walu{\pmvM}{0}{\tokT}$ be a state where the adversary has no tokens, and there are no contracts.
Consider a contract $\contract{Doubler}$ with a function that, upon receiving as input $n:\tokT$, returns to the sender $2n:\tokT$, and let
$\cstD = \walpmv{\contract{Doubler}}{\waltok{2}{\tokT}}$.
% Consider the following instance of the $\contract{Airdrop}$ being added to a blockchain state $\sysS$ with no contracts, where we want to deploy an $\contract{Exchange}$ contract like the one in~\Cref{lst:exchange}:
% \begin{align*}
%     \cstC = \walpmv{\contract{Airdrop}}{\waltok{1}{\tokT}}
%     \\
%     \cstD & = \walpmv{\contract{Exchange}}{\waltok{100}{\ETH}, \code{tin}=\tokT, \code{tout}=\ETH, \code{rate}=10, \code{owner}=\pmvB}
% \end{align*}
%
By~\Cref{lem:qnonint:basic}, $\sysS$ does not interfere with $\cstD$. 
Instead, adding $\cstC = \walpmv{\contract{Airdrop}}{\waltok{1}{\tokT}}$
to $\sysS$ yields $\qnonint{\sysS \mid \cstC}{\cstD} = 1$,
since $\lmev{\setenum{\contract{Doubler}}}{\sysS \mid \cstC \mid \cstD}{\setenum{\contract{Doubler}}} = 0$ while
$\lmev{}{\sysS \mid \cstC \mid \cstD}{\setenum{\contract{Doubler}}} = 2 \cdot \price{\tokT}$.
This increase is caused by the ability of $\pmvM$ to leverage the token dependencies between the newly deployed $\contract{Airdrop}$ contract to extract more MEV from $\contract{Doubler}$ than previously possible.
\hfill\qedex
\end{example}

The previous example also shows that wealthier adversaries not always cause greater interference.
Indeed, if $\sysS = \walu{\pmvM}{1}{\tokT}$, then $\pmvM$ does not need to exploit the $\contract{Airdrop}$ to extract MEV from the $\contract{Doubler}$ contract, since she has enough tokens in her wallet. 
Of course there are also cases where wealthier adversary can cause more MEV interference: we will see this in~\Cref{ex:amm-bet}, where a sufficiently wealthy $\pmvM$ can win a bet by producing a price fluctuation in an AMM.
%
% Therefore, MEV interference does not enjoy monotonicity with respect to the adversary's wealth in either direction.

\medskip
\Cref{prop:qnonint:adv wallets} shows that users' wallets are irrelevant to the evaluation of MEV interference. 
Namely, $\qnonint{\sysS}{\cstD}$ is preserved when removing from $\sysS$ all the wallets except those of adversaries. 
Recall that a wallet state $\WmvA$ is a map from accounts to wallets.
Then, in a state $\sysS = \WmvA \mid \cstC$, we just need to consider the restriction of $\WmvA$ to the domain $\Adv$.

\begin{theorem}
\label{prop:qnonint:adv wallets}
If $\dom{\WmvA[\Adv]} = \Adv$, then
$\qnonint{\WmvA[\Adv] \mid \WmvA \mid \cstC}{\cstD} = \qnonint{\WmvA[\Adv] \mid \cstC}{\cstD}$.
\end{theorem}

% $\qnonint{\sysS}{\cstD}$ should be independent of the users' wallets in $\sysS$, except for those belonging to adversaries.

Here the intuition is that the adversary does not have any control of the tokens in users' wallets, and therefore these tokens play no role in the extraction of MEV from $\cstD$. 
This assumption highlights a simplification in our attacker model, namely that the mempool of users' transactions is not known by the adversary. 
Formally, this assumption is visible in the definition of MEV in~\eqref{eq:lmev}, where the set $\mall{\CmvD}{\Adv}$ of transactions craftable by the adversary does not take the mempool as a parameter.
Were mempool transactions playable by the adversary, then their success would also depend on the users' wallet, and consequently the MEV interference would possibly depend on them. 
We discuss this in~\Cref{sec:conclusions}.

\medskip
\Cref{th:qnonint:preserving-interference} provides sufficient conditions under which an adversary $\Adv$ gains no advantage by front-running the newly deployed contracts~$\cstD$ with malicious contracts $\cstC[\Adv]$.
Condition~\ref{condition:qnonint:preserving-interference:1} requires the contracts in $\deps{\cstD}$ to be \emph{sender-agnostic}, \ie their functions are unaware of the identity of the sender, only being able to use it as a recipient of token transfers.
Condition~\ref{condition:qnonint:preserving-interference:2} requires  that the contracts in $\deps{\cstD}$ are \emph{token independent} with those in the other contracts (not in $\deps{\cstD}$) which could be possibly exploited by $\Adv$.
Note that since~\Cref{def:qnonint} assumes that states are well-formed, \Cref{th:qnonint:preserving-interference} implicitly assumes that contracts in $\cstD$ do not call contracts in $\cstC[\Adv]$.
Before stating \Cref{th:qnonint:preserving-interference},
we formalise 
% these conditions.
sender-agnosticism and token independence.

\begin{definition}[Sender-agnosticism]
A contract $\contract{C}$ is \emph{sender-agnostic} if, for all states $\sysS$ and for all transitions that involve an (external or internal) call
to $\contract{C}$, 
% $\addr{a}:\contract{C}.\txcode{f}(\code{args})$,
replacing the caller's address $\addr{a}$ with any other address $\addr{b}$ results in the same post-transition state, up to the substitution of $\addr{a}$ with $\addr{b}$.
\end{definition}

In practice, the effect of calling a function of a sender-agnostic contract $\contract{C}$ can be decomposed into:
\begin{inlinelist}

\item updating the states of contracts (either directly or through internal calls);

\item transferring tokens between users and contracts;

\item transferring tokens to the $\sender$ of the call to $\contract{C}$.

\end{inlinelist}
Any call to $\contract{C}$ with the same arguments and $\origin$, but distinct $\sender$, has exactly the same effect, except for item (iii), where tokens are transferred to the new sender.

Token independence relies on two auxiliary notions: the token types that can be received by contracts $\cstC$ from other contracts in $\sysS$, denoted by $\intok{\sysS}{\cstC}$, and those that can be sent from $\cstC$ to other contracts, written~$\outtok{\sysS}{\cstC}$.

\begin{definition}[Token independence]
\label{def:token-independence}
Let $\sysS = \WmvA \mid \cstC$, and let $\cstD \preceq \cstC$ be a subset of the contract states in $\cstC$.
We define: 
\begin{itemize}

\item $\intok{\sysS}{\cstD}$ as the set of token types $\tokT$ for which
there exists a state $\sysSi$ reachable from $\sysS$ through a sequence of steps, containing a transaction that causes an inflow of tokens $\tokT$ from outside $\cstD$ to one of the contracts in $\cstD$.

\item $\outtok{\sysS}{\cstD}$ as the set of token types $\tokT$ for which there exists a state $\sysSi$ reachable from $\sysS$ through a sequence of steps, containing a transaction that causes an outflow of tokens $\tokT$ from $\cstD$ to one of the contracts outside $\cstD$.

\end{itemize}
    
\noindent
Let now $\cstD[0] \preceq \cstC$ and $\cstD[1] \preceq \cstC$.
We say that $\cstD[0]$ and $\cstD[1]$ are \emph{token independent} in $\sysS = \WmvA \mid \cstC$ when $\intok{\sysS}{\cstD[0]} \cap \outtok{\sysS}{\cstD[1]} = \emptyset = \intok{\sysS}{\cstD[1]} \cap \outtok{\sysS}{\cstD[0]}$.
\end{definition}

\begin{theorem}
\label{th:qnonint:preserving-interference}
\mbox{$\qnonint{\sysS}{\cstD} = \qnonint{\sysS \mid \cstC[\Adv]}{\cstD}$}
holds if
\begin{inlinelist}

\item \label{condition:qnonint:preserving-interference:1}
the contracts in $\deps{\cstD}$ are sender-agnostic, and

\item \label{condition:qnonint:preserving-interference:2}
$\deps{\cstD}$ and
$\deps{\sysS \mid \cstC[\Adv]} \setminus \deps{\cstD}$
are token independent in $\sysS \mid \cstC[\Adv] \mid \cstD$.

\end{inlinelist}
\end{theorem}

Note that $\qnonint{\sysS}{\cstD}$ is zero when the contract dependencies and the token dependencies of $\cstD$ in $\sysS$ are irrelevant to the ability of inflicting a loss to $\cstD$.
\Eg, consider an arbitrary state $\sysS$ where we want to deploy an airdrop contract $\cstD$ (see~\Cref{lst:airdrop}). 
In this scenario, the adversary cannot gain any advantage from the contracts in $\sysS$, since she can extract the full MEV from the airdrop by interacting with $\cstD$, only. 
Therefore, the MEV interference from $\sysS$ to $\cstD$ is zero. 
\section{Use cases}
\label{sec:use-cases}

We now illustrate MEV interference through a set of use cases.
For simplicity, we assume the values in these use cases are real numbers and that all computations are performed using exact arithmetic.
We note that adapting our results to smart contract platforms, that, like Ethereum, operate on integers, requires several modifications, such as applying flooring to arithmetic operations and replacing equalities with inequalities.
We refer to \iftoggle{arxiv}{\Cref{sec:proofs:use-cases}}{the full technical report on ArXiV} for details.

\begin{figure}[t]
\begin{lstlisting}[
  ,language=txscript
  ,morekeywords={Exchange,swap,getRate,setRate,getTokens},classoffset=4
  ,morekeywords={a,b,A,Oracle},keywordstyle=\pmvColor,classoffset=4
  ,morekeywords={t,t1,t2,T,tin,tout},keywordstyle=\tokColor
  ,caption={An exchange contract.}
  ,label={lst:exchange}
]
contract (*$\contract{Exchange}$*) {
  constructor(pay x:tout_, tin_, rate_) { // receive x:tout_ from sender
    require rate_>0 && tin_!=tout_; 
    rate=rate_; tout=tout_; tin=tin_; owner=sender; 
  }
  getTokens() { return (tin,tout); }
  getRate(tout) { return rate; } // 1:tin for getRate(tout):tout
  setRate(r) { require sender==owner; rate=r; } // sender can update rate 
  
  swap(pay x:tin) {           // sender sells x units of token tin
    y = x*getRate(tout);      // units of token tout sold to sender  
    require balance(tout)>=y; // Exchange has enough tout tokens
    transfer(sender, y:tout); // send y units of token tout to sender 
  }
}
\end{lstlisting}
\end{figure}

\begin{example}[Airdrop/Exchange]
\label{ex:airdrop-exchange}
Consider an instance of the $\contract{Exchange}$ contract in \Cref{lst:exchange}, to be deployed in a
blockchain state $\sysS$ containing an instance of the $\contract{Airdrop}$ contract in \Cref{lst:airdrop}.
More specifically, let: 
\begin{align*}
    \sysS & = \walu{\pmvM}{n_{\pmvM}}{\tokT} \mid \walpmv{\contract{Airdrop}}{\waltok{n_{\contract{A}}}{\tokT}}
    \\
    \cstD & = \walpmv{\contract{Exchange}}{\waltok{n_{\contract{E}}}{\ETH}, \code{tin}=\tokT, \code{tout}=\ETH, \code{rate}=r, \code{owner}=\pmvA}
\end{align*}
The $\contract{Exchange}$ contract allows any user to swap tokens of type $\code{tin}$ with tokens of type $\code{tout}$ (in the instance, $\tok{T}$ and $\ETH$, respectively), at an exchange rate of 1 unit of $\code{tin}$ for $\code{rate}$ units of $\code{tout}$.
For simplicity, assume that 
% $\code{rate}$ is a floating-point number, and arithmetic operations are floored, and that 
$\price{\tokT} = \price{\tok{ETH}} = 1$. 
We evaluate the MEV interference from $\sysS$ to $\cstD$.
When the exchange rate is favourable, \ie $r > 1$, the adversary $\pmvM$ can extract MEV from $\cstD$ by exchanging $\tokT$ for $\ETH$. 
This is possible as far as $\contract{Exchange}$ has enough $\ETH$ balance.
The MEV can be further increased by draining $n_{\contract{A}}:\tokT$ from $\contract{Airdrop}$, and swapping these tokens through the $\contract{Exchange}$.
More precisely, we have:
\begin{align*}
    \lmev{\setenum{\contract{Exchange}}}{\sysS \mid \cstD}{\setenum{\contract{Exchange}}}
    & = \begin{cases}
        n_{\pmvM} \cdot r & \text{if}\ n_{\pmvM} < \nicefrac{n_{\contract{E}}}{r} \\
        n_{\contract{E}} & \text{otherwise}
      \end{cases}
    \\
    \lmev{}{\sysS \mid \cstD}{\setenum{\contract{Exchange}}} 
    & = \begin{cases}
        (n_{\pmvM} + n_{\contract{A}}) \cdot r & \text{if}\  n_{\pmvM}  < \nicefrac{n_{\contract{E}}}{r} - n_{\contract{A}} \\
        n_{\contract{E}} & \text{otherwise}
      \end{cases}
\end{align*}
Therefore, the MEV interference from $\sysS$ on $\cstD$ is given by: 
% \emnote{Yes Sir, I am checking it again now. Sir, the actual interference value is 1 minus the ratio of the two MEVs (which are in floor). In the first case, the interference value will be bounded by $\nicefrac{n_{\contract{A}}}{(n_{\pmvM} + n_{\contract{A}})}$}
\begin{align*}
    \qnonint{\sysS}{\cstD} 
    \; = \;
    \begin{cases}
        %\emnote{1 - \frac{\lfloor \nicefrac{n_{\pmvM} \cdot r}{100} \rfloor}{\lfloor \nicefrac{(n_{\pmvM} + n_{\contract{A}}) \cdot r}{100} \rfloor}}
        \nicefrac{n_{\contract{A}}}{(n_{\pmvM} + n_{\contract{A}})}
        & \text{if}\ n_{\pmvM} < \nicefrac{n_{\contract{E}}}{r} - n_{\contract{A}}
        \\
        1 - \nicefrac{n_{\pmvM} \cdot r}{n_{\contract{E}}}
        &\text{if}\ \nicefrac{n_{\contract{E}}}{r} - n_{\contract{A}} \leq n_{\pmvM} < \nicefrac{n_{\contract{E}}}{r} 
        % (n_{\pmvM} + n_{\contract{A}}) \cdot r \geq n_{\contract{E}} > n_{\pmvM} \cdot r 
        \\
        0 & \text{otherwise}
      \end{cases}
\end{align*}
When $\pmvM$ is sufficiently rich, she can drain the $\contract{Exchange}$ without invoking the $\contract{Airdrop}$.
Instead, when $\pmvM$'s wealth is limited, she is able to inflict a greater loss of $\contract{Exchange}$ by leveraging the $\contract{Airdrop}$.
So, the interference caused to $\contract{Exchange}$ in this case has a dual dependence on the adversary's and the $\contract{Airdrop}$'s wealth.
Furthermore, the interference is inversely proportional to $\pmvM$'s wealth, \ie richer adversaries have less need to exploit the context, resulting in lower interference from $\sysS$ to $\cstD$.
This is coherent with our intuition, since we would expect a poorer adversary to benefit more from exploiting the $\contract{Airdrop}$ than a richer one.
\hfill\qedex
\end{example}

\begin{example}[AMM/Bet]
\label{ex:amm-bet}
The $\contract{Bet}$ contract in~\Cref{lst:bet} allows a player to bet on the exchange rate between a token and $\ETH$.
It is parameterized over an $\contract{oracle}$ that is queried for the token price. 
To enter the bet, the player must match the initial pot set upon deployment.
Before the deadline, the player can win a fraction $\code{potShare}$ of the pot if the oracle exchange rate exceeds or equals $\code{potShare}$ times the rate.
The remaining fraction is taken by the owner.
Consider an instance of $\contract{Bet}$ using the $\contract{AMM}$ in~\Cref{fig:amm} as a price oracle:
\begin{align*}    
\sysS & = \walu{\pmvM}{m}{\ETH} \mid
\walpmv{\contract{AMM}}{\waltok{r_0}{\ETH},\waltok{r_1}{\tokT}} \mid
\code{block.num} = d-k \mid \cdots
\\
\cstD & = \walpmv{\contract{Bet}}{\waltok{b}{\ETH}, \code{owner}=\pmvA, \code{tok}=\tokT, \code{rate}=r, \code{deadline}=d}
\end{align*}
When $\pmvM$ is allowed to leverage $\contract{Bet}$'s dependency, she can manipulate the $\contract{AMM}$ to influence the internal exchange rate.
If $\pmvM$ has sufficient funds to enter the bet, she can fire the following sequence of transactions, where, in the $\txcode{swap}$ transaction, $x = m-b \geq 0$ is the number of $\ETH$ units sent to the $\contract{AMM}$ and $y = \nicefrac{x r_1}{r_0 + x}$ is the number of $\tokT$ units received
(we omit $\pmvM$'s wallet for brevity):

\begin{align*}
    \sysS \mid \cstD
    & \xrightarrow{\pmvM:\contract{Bet}.\txcode{bet}(\pmvM\ \code{pays}\ \waltok{b}{\ETH},p)}
    &&
    % \walu{\pmvM}{x}{\ETH} \mid
    \walpmv{\contract{AMM}}{\waltok{r_0}{\ETH},\waltok{r_1}{\tokT}} \mid
    \walpmv{\contract{Bet}}{\waltok{2b}{\ETH},\code{potShare}=p,\cdots} \mid \cdots
    \\
    & \xrightarrow{\pmvM:\contract{AMM}.\txcode{swap}(\pmvM\ \code{pays}\ \waltok{x}{\ETH}, 0)}
    &&
    % \walu{\pmvM}{y}{\tokT} \mid
    \walpmv{\contract{AMM}}{\waltok{r_0 + x}{\ETH},\waltok{r_1 - y}{\tokT}} \mid \walpmv{\contract{Bet}}{\waltok{2b}{\ETH},\cdots} \mid \cdots
    \\
    & \xrightarrow{\pmvM:\contract{Bet}.\txcode{win}()}
    &&
    % \walpmv{\pmvM}{\waltok{2bp}{\ETH},\waltok{y}{\tokT}} \mid
    \walpmv{\contract{AMM}}{\waltok{r_0 + x}{\ETH},\waltok{r_1 - y}{\tokT}} \mid \walpmv{\contract{Bet}}{\waltok{2b- 2bp}{\ETH},\cdots} \mid \cdots
    \\
    & \xrightarrow{\pmvM:\contract{AMM}.\txcode{swap}(\pmvM\ \code{pays}\ \waltok{y}{\tokT}, 0)}
    &&
    % \walu{\pmvM}{2bp + x}{\ETH} \mid
    \walpmv{\contract{AMM}}{\waltok{r_0}{\ETH},\waltok{r_1}{\tokT}} \mid
    \walpmv{\contract{Bet}}{\waltok{2b- 2bp}{\ETH},\cdots} \mid \cdots
\end{align*}
The bet value that maximizes the loss caused to $\contract{Bet}$ depends on $\pmvM$'s wealth, and is given by $p = \nicefrac{r_0 + x}{r(r_1 - y)}$.
Assuming $\pmvM$ enters the bet only for $p \geq \nicefrac{1}{2}$ (since a smaller proportion makes the bet irrational for her), by~\Cref{eq:lmev:unrestricted} we have:
\begin{align*}
    \lmev{}{\sysS \mid \cstD}{\setenum{\contract{Bet}}} 
    % \leq
    =
    \left( \nicefrac{2 (r_0 + m - b)^2}{r r_0 r_1} - 1 \right) b
\end{align*}
If $\pmvM$ can only interact with $\contract{Bet}$, she is limited to settle on a lower bet value: 
% and~\Cref{eq:lmev} gives:
\begin{align*}
    \lmev{\setenum{\contract{Bet}}}{\sysS \mid \cstD}{\setenum{\contract{Bet}}}
    % >
    =
    \begin{cases}
        \nicefrac{2br_0}{rr_1} -b -1 & \text{if}\  \nicefrac{r_0}{r r_1} \geq \nicefrac{1}{2} \\
        0 & \text{otherwise}
    \end{cases}
\end{align*} 
Accordingly, $\lmev{}{}{}$ interference is estimated through \Cref{def:qnonint} as follows:
\begin{align*}
    \qnonint{\sysS}{\cstD}
    % & <
    & = 
    \begin{cases}
        1 - \frac{2br_0^2 - rr_0r_1(b+1)}{2b(r_0+m-b)^2 - brr_0r_1}
        & \text{if}\ \nicefrac{r_0}{r r_1} \geq \nicefrac{1}{2} \\
        1 & \text{otherwise}
    \end{cases}
\end{align*}
We observe maximum interference when $\pmvM$ exploits the $\contract{Bet}$ by manipulating the $\contract{AMM}$, which would be impossible by interacting exclusively with $\contract{Bet}$.
Furthermore, the interference value is proportional to the adversarial wealth, as one would anticipate.
By contrast, even if $\pmvM$ was able to empty a portion of the $\contract{Bet}$ by fair play, she can always increase this loss by manipulating the $\contract{AMM}$ (provided she owns adequate funds).
% Thus, the amount of interference relies on the reserves of the $\contract{AMM}$, the adversarial wealth and the bet rate.
Note that in the composition between $\contract{Bet}$ and $\contract{Exchange}$, the MEV interference is zero, as the adversary cannot manipulate the exchange rate (unless she is the $\contract{Exchange}$ owner).
\hfill\qedex
\end{example}

\begin{figure}[t!]
\begin{lstlisting}[
,language=txscript
,morekeywords={bet,win,close,getTokens,getRate},classoffset=4
,morekeywords={a,b,A},keywordstyle=\pmvColor,classoffset=5
,morekeywords={t,ETH},keywordstyle=\tokColor,classoffset=6
,morekeywords={Bet,oracle},keywordstyle=\cmvColor
,caption={A Bet contract.}
,label={lst:bet}
]
contract (*$\contract[oracle]{Bet}$*) {
  constructor(pay x:ETH, tok_, deadline_, rate_) {
    require tok_!=ETH && oracle.getTokens()==(ETH,tok_);
    tok=tok_; deadline=deadline_; rate=rate_; owner=sender;
  }
  bet(pay x:ETH, p_) { // sender gives x:ETH to Bet and chooses potShare
    require player==null && x==balance(ETH) && p_>=0 && p_<=1;
    potShare = p_; player=sender;
  }
  win() { // only callable by player before the deadline 
    require block.num<=deadline && sender==player;
    if (oracle.getRate(ETH)>=potShare*rate)
      transfer(player, potShare*balance(ETH):ETH);
  }
  close() { // after the deadline, transfer the ETH balance to the owner
    require block.num>deadline;
    transfer(owner, balance(ETH):ETH);
  }
}
\end{lstlisting}
% \end{figure}
%
% \begin{figure}[t!]
\begin{lstlisting}[
,language=txscript
,morekeywords={AMM,addLiq,swap,getTokens,getRate},classoffset=4
,morekeywords={T0,T1,t,tin,tout},keywordstyle=\tokColor,classoffset=5
,caption={A constant-product AMM contract.}
,label={fig:amm}
]
contract (*$\contract{AMM}$*) {
  constructor(pay x0:T0, pay x1:T1) { require x0>0 && x1>0; }
  
  getTokens() { return (T0,T1); }   // token pair
 
  getRate(tout) { // 1:tin for getRate(tout):tout
    if (tout==T0) { tin=T1 } else { tin=T0 };
    return balance(tout)/balance(tin);
  }
  swap(pay x:tin, ymin) { // sell x:tin to buy at least ymin:tout
    if (tin==T0) { tout=T1 } else { tout=T0 };
    y = x*getRate(tout);        // units of token tout sold to sender  
    require ymin<=y<balance(tout);  // the AMM has enough tout tokens
    transfer(sender, y:tout); // send y units of token tout to sender
  } 
}
\end{lstlisting}
\end{figure}

\begin{example}[AMM/Lending Pool]
\label{ex:amm-LP}
The contract $\contract{LP}$ in~\Cref{lst:LP} implements a simplified lending protocol, where users can deposit and borrow tokens. 
Borrowing requires users to have a sufficient \emph{collateralization}~\cite{Gudgeon20aft,BCL21wtsc}. 
This value, defined as the ratio between the value of their deposits and that of their debits, is a measure of the borrowing capacity (full versions of lending protocols include a function that allows liquidators to repay loans of under-collateralized borrowers in exchange for part of their collateral). 
The contract $\contract{LP}$ is parameterized over an $\contract{oracle}$ that is queried for the token prices.
Below we analyze a well-known attack where the underlying $\contract{oracle}$ is an $\contract{AMM}$, which is manipulated by an adversary to increase her  borrowing capacity~\cite{Gudgeon2020cvcbt,Qin21fc,BCL21wtsc,Mackinga22icbc,Arora24asiaccs}.

More specifically, consider the following instance, where $\price{\ETH} = 1 = \price{\tokT}$, the $\contract{AMM}$ is balanced, and the adversary $\pmvM$ has not deposited or borrowed tokens yet:
\begin{align*}    
\sysS = \walu{\pmvM}{n}{\ETH} \mid
\walpmv{\contract{AMM}}{\waltok{r}{\ETH},\waltok{r}{\tokT}}
&&\cstD = \walpmv{\contract{LP}}{\waltok{a}{\ETH}, \waltok{b}{\tokT}, \code{Cmin} = C_{\it min}, \cdots}
\end{align*}

If $\pmvM$ can interact with the $\contract{AMM}$, she has the following attack strategy: 
deposit $(n-x):\ETH$ to the $\contract{LP}$, and use the remaining $x:\ETH$ to inflate the price of $\tokT$ in the $\contract{AMM}$.
This allows $\pmvM$ to increase the amount of $\tokT$ she can borrow, since the $\contract{LP}$ now uses an artificially inflated price to determine her borrowing capacity.

To implement this strategy, $\pmvM$ fires the following sequence of transactions,
where we denote by $y$ the amount of $\tokT$ units that $\pmvM$ receives from the $\txcode{swap}$, and with $t$ the amount of $\tokT$ units that $\pmvM$ manages to borrow from the $\contract{LP}$
(below, we omit $\pmvM$'s wallet, and the parts of the state that do not change upon a transition): 
\begin{align*}
    \sysS \mid \cstD
    & \xrightarrow{\pmvM:\contract{LP}.\txcode{deposit}(\pmvM\ \code{pays}\ \waltok{(n-x)}{\ETH})}
    &&
    % \walu{\pmvM}{x}{\ETH} \mid
    \walpmv{\contract{AMM}}{\waltok{r}{\ETH},\waltok{r}{\tokT}}
    \mid
    \walpmv{\contract{LP}}{\waltok{a + n - x}{\ETH}, \waltok{b}{\tokT}, \cdots} \mid \cdots
    \\
    & \xrightarrow{\pmvM:\contract{AMM}.\txcode{swap}(\pmvM\ \code{pays}\ \waltok{x}{\ETH}, 0)}
    &&
    % \walu{\pmvM}{y}{\tokT} \mid
    \walpmv{\contract{AMM}}{\waltok{r + x}{\ETH},\waltok{r - y}{\tokT}} 
    % \mid \walpmv{\contract{LP}}{\waltok{a + n - x}{\ETH},\waltok{b}{\tokT}, \cdots} 
    \mid \cdots
    \\
    & \xrightarrow{\pmvM:\contract{LP}.\txcode{borrow}(t,\tokT)}
    &&
    % \walu{\pmvM}{y + t}{\tokT} \mid
    % \walpmv{\contract{AMM}}{\waltok{r + x}{\ETH},\waltok{r - y}{\tokT}} \mid
    \cdots \mid 
    \walpmv{\contract{LP}}{\waltok{a + n - x}{\ETH},\waltok{b - t}{\tokT}, \cdots} \mid \cdots
    \\
    & \xrightarrow{\pmvM:\contract{AMM}.\txcode{swap}(\pmvM\ \code{pays}\ \waltok{y}{\tokT}, 0)}
    &&
    % \walpmv{\pmvM}{\waltok{x}{\ETH},\waltok{t}{\tokT}} \mid
    \walpmv{\contract{AMM}}{\waltok{r}{\ETH},\waltok{r}{\tokT}}
    % \mid \walpmv{\contract{LP}}{\waltok{a + n - x}{\ETH},\waltok{b - t}{\tokT}, \cdots} 
    \mid \cdots
\end{align*}

\noindent
The amount that $\pmvM$ can borrow (as a function of $x$) is 
$t = \nicefrac{(n-x) (r + x)^2}{C_{\it min} (r - y)^2}$.
Its maximum is obtained for $x = \nicefrac{4n-r}{5}$
when $\pmvM$ benefits from the manipulation (\ie, when $4n \geq r$), 
and for $x=0$ otherwise.
% $y = \nicefrac{x r}{r + x}$.

Assuming that the $\contract{LP}$ has sufficient funds, the unrestricted MEV is given by:
% by~\Cref{eq:lmev:unrestricted} we get:
\begin{align*}
    \lmev{}{\sysS \mid \cstD}{\setenum{\contract{LP}}} 
    % = t + x - n
    & \; = \; \frac{(n-x)(r+x)^2}{C_{min}(r-y)^2} + x - n
    \\
    & \; = \; \begin{cases}
        \left( \frac{n+r}{5} \right) \left( \frac{1}{C_{min}} \left( \frac{4(n+r)}{5r}\right)^4 - 1 \right) \;
        & \text{if}\ 4n \geq r
        \\
        n \left( \frac{1}{C_{min}} - 1 \right) & \text{otherwise}
        \end{cases}
\end{align*}
On the contrary, if $\pmvM$ was restricted to interact with the $\contract{LP}$ only, she suffers a reduced borrowing allowance.
By~\Cref{eq:lmev} we have:
\begin{align*}
    \lmev{\setenum{\contract{LP}}}{\sysS \mid \cstD}{\setenum{\contract{LP}}}
    = n \left( \frac{1}{C_{min}} - 1 \right) 
    % \frac{n}{C_{\it min}} - n
\end{align*} 
Accordingly, $\lmev{}{}{}$ interference is estimated through \Cref{def:qnonint} as follows:
\begin{align*}
    \qnonint{\sysS}{\cstD}
    = \begin{cases}
        1 \; - \; \frac{5^5 r^4 n(1-C_{min})}{(n+r)\left( 4^4(n+r)^4 - (5r)^4C_{min}\right)} \;
        & \text{if}\ 4n \geq r \\
        0 \; & \text{otherwise}
    \end{cases}
\end{align*}
In accordance with our expectations, the interference is indeed proportional to the attack capital $n$ of the adversary.
Naturally, adversaries with higher manipulation capital experience an increased borrowing capacity.
Moreover, the degree of interference is influenced by the $\contract{AMM}$ reserves since the profitability of the attack rests on the cost of manipulating and de-manipulating the $\contract{AMM}$.
\hfill\qedex
\end{example}

\begin{figure}[t]
\begin{lstlisting}[
  ,language=txscript
  ,morekeywords={LP,collateral,deposit,borrow,getTokens,getRate},classoffset=4
  ,morekeywords={a,b,A,Oracle},keywordstyle=\pmvColor,classoffset=4
  ,morekeywords={t,t1,t2,T,T0,T1},keywordstyle=\tokColor
  ,caption={A Lending Pool contract (simplified).}
  ,label={lst:LP}
]
contract (*$\contract[oracle]{LP}$*) {
  constructor(Cmin_) { Cmin = Cmin_; } // collateralization threshold
  
  collateral(a) { // return a's collateralization
    val_minted = 0;
    for c in minted: val_minted += minted[t][a] * (*$\contract{oracle}$*).getRate(t);
    val_debts = 0;
    for c in debts:  val_debts  += debt[t][a] * (*$\contract{oracle}$*).getRate(t);
    return val_minted / val_debts; 
  } 
  deposit(a pays x:t) { // a deposits x units of token t in the LP
    minted[t][a] += x;  // record the deposited units in the minted map
  }
  borrow(a sig, x, t) { // a borrows x units of token t in the LP
    require balance(t)>=x;
    debts[t][a] += x;   // record the borrowed units in the debts map
    require collateral(a)>=Cmin; // a is over-collateralized
    transfer(a, x:t);
  }
}
\end{lstlisting}
\end{figure}
\section{Conclusions}
\label{sec:conclusions}

We have proposed a notion of economic security for smart contract compositions, which quantifies the potential economic loss an adversary can inflict on a contract by targeting its dependencies. Below, we discuss some limitations of our approach and directions for future work.

\paragraph{Limitations}
To keep our theory manageable, we have made a few simplifying assumptions in our model.
A first assumption is that the prices of native crypto-assets are constant. Consequently, the amount of MEV interference is not affected by fluctuations of these prices (while they could depend on the prices 
provided by DEXes, like in~\Cref{ex:amm-bet,ex:amm-LP}).
Handling price updates would require to extend blockchain states with a function mapping tokens to their prices.
Another assumption is that the local MEV in~\Cref{eq:lmev} does not allow adversaries to exploit their knowledge of pending users' transactions (the public \emph{mempool}).
The rationale underlying this choice is that, in our vision, MEV interference should be the basis for a static analysis of smart contracts, where dynamic data such as the mempool transactions are not known.
Assuming an over-approximation of users' transactions, we could extend our MEV interference by making the mempool a parameter of local MEV, similarly to what done for the theory of MEV in~\cite{BZ25fc}.

\paragraph{Future work} 
While some tools exist for detecting price manipulation attacks in DeFi protocols~\cite{Wu21defiranger,Kong23defitainter,KaWaiWu25flashdefier}, and others for estimating MEV opportunities~\cite{Babel23clockwork,Babel23ccs}, 
there remains a gap in addressing general economic attacks on smart contract compositions.
A common analysis technique underlying the detection of price manipulation attacks --- also employed by some of the tools mentioned above --- is \emph{taint analysis}, which aims at identifying potential data flows from low-level to high-level data.
In the DeFi setting, this typically corresponds to flows from to functions that influence token prices to functions that transfer tokens.
While this technique could potentially be generalised to analyse \emph{qualitative} MEV non-interference, capturing our notion of \emph{quantitative} interference seems to require more advanced techniques. 
Some inspiration could be drawn from static analysis techniques for information-theoretic interference~\cite{Clark07jcs,Kopf13sfm,Klebanov14tcs,Assaf17popl}.
We plan to explore this research line in future work.
Our blockchain model represents crypto-assets as token types with primitive transfer operations and built-in linearity guarantees preventing asset creation or destruction. 
In practice, several blockchains including Ethereum do not have native support for custom tokens, but rather require to implement 
them as smart contracts exposing standard interfaces.
This opens the door for attackers to exploit potential discrepancies between these implementations and the standards, possibly leading to MEV~\cite{Chen19ccs}.
Applying our MEV interference analysis to such compositions is left as future work.

\iftoggle{anonymous}{}{\paragraph*{Acknowledgments}

Work partially supported by project SERICS (PE00000014)
under the MUR National Recovery and Resilience Plan (NRRP) funded by the European Union -- NextGenerationEU, and by PRIN 2022 NRRP project DeLiCE (F53D23009130001).}

\bibliographystyle{splncs04}
\bibliography{main}

\iftoggle{arxiv}{%
\newpage
\appendix
\section{Proofs: properties of MEV interference} 
\label{sec:proofs}

We start by recalling from~\cite{BMZ24fc} a few useful properties of local MEV. 
We define the relation $\preceq$ between contract states as follows:
\[
\cstC \preceq \cstD
\quad\iff\quad
\forall \contract{C} \in \dom{\cstC}.\
\contract{C} \in \dom{\cstD}
\,\land\,
\cstC(\contract{C}) = \cstD(\contract{C})
\]
% we use the notation $\cstD |_{\CmvC}$ to denote the state obtained by removing from $\cstD$ all the states of contracts not in $\CmvC$. 
Therefore, the condition $\cstC \preceq \cstD$ in~\Cref{lem:lmev:monotonicity} of~\Cref{lem:lmev} means that $\cstD$ is a widening of the state $\cstC$ with other arbitrary contract states. 

\begin{lemma}[Basic properties of MEV~\cite{BMZ24fc}]
  \label{lem:lmev}
  For all $\sysS$, $\CmvC,\CmvD \subseteq \CmvU$:
  \begin{enumerate}

  \item \label{lem:lmev:mev}
    $\lmev{\CmvD}{\sysS}{\emptyset} = \lmev{\emptyset}{\sysS}{\CmvC} = 0$,
    $\lmev{\CmvU}{\sysS}{\CmvU} \geq \mev{}{\sysS}{}$

  \item \label{lem:lmev:L-leq-H}
    if $\CmvD \subseteq \CmvDi$, then
    $\lmev{\CmvD}{\sysS}{\CmvC} \leq \lmev{\CmvDi}{\sysS}{\CmvC}$

  \item \label{lem:lmev:monotonicity}
    $\lmev{\CmvD}{\WmvA \mid \cstC}{\CmvC} \leq \lmev{\CmvD}{\WmvA \mid \cstD}{\CmvC}$ if $\cstC \preceq \cstD$

  \item \label{lem:lmev:garbage}
    $\lmev{\CmvD}{\WmvA \mid \cstC}{\CmvC} = \lmev{\CmvD}{\WmvA \mid \cstC}{\CmvC \cap \cmvOfcst{\cstC}} = \lmev{\CmvD\cap \cmvOfcst{\cstC}}{\WmvA \mid \cstC}{\CmvC}$
    
  \item \label{lem:lmev:leq-wealth}
    $0 \leq \lmev{\CmvD}{\sysS}{\CmvC} \leq \wealth{\CmvC}{\sysS}$
    
  \end{enumerate}
\end{lemma}

\Cref{lem:lmev-wallet} states that the only user wallets that
need to be taken into account to estimate the  MEV are those of the adversary.
This is because $\Adv$ has no way to force other users to spend their
tokens in the attack sequence.%

\begin{lemma}[MEV and adversaries' wallets~\cite{BMZ24fc}]
  \label{lem:lmev-wallet}
  If $\dom{\WmvA[\Adv]} = \Adv$, then
  \[
    \lmev{\CmvD}{\WmvA[\Adv] \mid \WmvA \mid \cstC}{\CmvC} = \lmev{\CmvD}{\WmvA[\Adv] \mid \cstC}{\CmvC}
  \]
\end{lemma}

\proofof{lem:qnonint:basic}
For Item~\ref{lem:qnonint:basic:1},
by~\Cref{lem:lmev:mev} of~\Cref{lem:lmev} we have that
$\lmev{}{\sysS\mid\emptyset}{\cmvOfcst{\emptyset}} = 0$.
The thesis follows by~\Cref{def:qnonint}.

\medskip\noindent
For Item~\ref{lem:qnonint:basic:2}, 
by \Cref{lem:lmev:garbage} of~\Cref{lem:lmev} we have: 
\[
    \lmev{}{\WmvA \mid \emptyset \mid \cstD}{\cmvOfcst{\cstD}} = \lmev{\cmvOfcst{\cstD}}{\WmvA \mid \emptyset \mid \cstD}{\cmvOfcst{\cstD}}
\]
which gives us $\qnonint{\WmvA \mid \emptyset}{\cstD} = 0$, and hence we have our thesis.

\medskip\noindent
For Item~\ref{lem:qnonint:basic:3}, there are two cases. If $\lmev{}{\sysS\mid\cstD}{\cmvOfcst{\cstD}} = 0$, 
then 
% by~\Cref{lem:lmev:L-leq-H} of~\Cref{lem:lmev}
% we also have 
% $\lmev{\cmvOfcst{\cstD}}{\sysS\mid\cstD}{\cmvOfcst{\cstD}} = 0$,
$\qnonint{\sysS}{\cstD} = 0$ holds by definition.
Otherwise, by~\Cref{lem:lmev:leq-wealth,lem:lmev:L-leq-H} of~\Cref{lem:lmev}:
\begin{align*}
    &0 
    \leq \lmev{\cmvOfcst{\cstD}}{\sysS \mid \cstD}{\cmvOfcst{\cstD}}
    \leq \lmev{}{\sysS \mid \cstD}{\cmvOfcst{\cstD}}
    \\
    \implies &0 
    \leq \frac{\lmev{\cmvOfcst{\cstD}}{\sysS \mid \cstD}{\cmvOfcst{\cstD}}}{\lmev{}{\sysS \mid \cstD}{\cmvOfcst{\cstD}}}
    \leq 1
    \\
    \implies &0 
    \leq 1 - \frac{\lmev{\cmvOfcst{\cstD}}{\sysS \mid \cstD}{\cmvOfcst{\cstD}}}{\lmev{}{\sysS \mid \cstD}{\cmvOfcst{\cstD}}}
    \leq 1
\end{align*}
which implies $0 \leq \qnonint{\sysS}{\cstD} \leq 1$, giving us our thesis.
\qed

\proofof{lem:qnonint:zero-wealth}
From~\Cref{lem:lmev:leq-wealth,lem:lmev:L-leq-H} of~\Cref{lem:lmev}, we have:
\begin{align*}
    0 \leq \lmev{\cmvOfcst{\cstD}}{\sysS\mid\cstD}{\cmvOfcst{\cstD}} \leq \lmev{}{\sysS\mid\cstD}{\cmvOfcst{\cstD}} \leq \wealth{\cmvOfcst{\cstD}}{\cstD}
\end{align*}
By hypothesis, $\wealth{\cmvOfcst{\cstD}}{\cstD} = 0$. So, by the inequalities above, $\lmev{}{\sysS\mid\cstD}{\cmvOfcst{\cstD}} = 0$. 
\Cref{def:qnonint} gives the thesis.
\qed

\begin{definition}[Gain]
\label{def:gain}
The gain of $\CmvC \subseteq \CmvU$ when
a transaction sequence $\TxTS$ is fired in $\sysS$ is given by
\(
\gain{\CmvC}{\sysS}{\TxTS}
=
\wealth{\CmvC}{\sysSi} - \wealth{\CmvC}{\sysS}
\)
if $\sysS \xrightarrow{\TxTS} \sysSi$.

\medskip\noindent
Dually, the loss of $\CmvC \subseteq \CmvU$ when a transaction sequence $\TxTS$ is fired in $\sysS$ is given by
  \(
  -\gain{\CmvC}{\sysS}{\TxTS}
  =
  \wealth{\CmvC}{\sysS} - \wealth{\CmvC}{\sysSi}
  \)
if $\sysS \xrightarrow{\TxTS} \sysSi$.
\end{definition}

\medskip
\Cref{lem:lmev:uniform} states that widening the contract state $\cstC$ preserves the MEV extractable from the target contracts.
This is because the contracts allowed to be targeted by the adversary, \ie $\CmvD$, are not widened.
% Note that,
% % for all $\sysS, \CmvC, \CmvD \subseteq \CmvU$, writing 
% whenever we write
% $\lmev{\CmvD}{\sysS}{\CmvC}$,
% we are implicitly assuming $\CmvD \subseteq \cmvOfcst{\sysS}$,
% \ie, the set of contracts callable by the adversary is always a subset of the contracts in the blockchain state.
% More explicitly, we do not give consideration to a possible temporal difference between deployment of contracts $\cstC$ and $\cstD$.
This refines~\Cref{lem:lmev:monotonicity} of~\Cref{lem:lmev}, giving an equality under the additional assumption $\CmvD \subseteq \cmvOfcst{\cstC}$.

\begin{lemma}
\label{lem:lmev:uniform}
\label{lem:lmev:prepending-state}
\label{lem:lmev:monotonicity-eq}
$\lmev{\CmvD}{\WmvA \mid \cstC}{\CmvC} = \lmev{\CmvD}{\WmvA \mid \cstD}{\CmvC}$
when $\CmvD \subseteq \cmvOfcst{\cstC}$
and \mbox{$\cstC \preceq \cstD$}.
\end{lemma}
\begin{proof}
    The inequality $\leq$ follows directly from~\Cref{lem:lmev:monotonicity} of~\Cref{lem:lmev}.
    For the inequality $\geq$, 
    assume that $\cstD$ is the composition of the contracts $\cstC$ with some other contracts $\bar{\cstC}$, \ie 
    $\cstC \preceq \cstD$, $\bar{\cstC} \preceq \cstD$,
    and $\cstD \preceq \cstC \mid \bar{\cstC}$.
    % % Write this as $\cstCi = \cstC \uplus \cstD$.
    %
    Let $\TxTS \in \mall{\CmvD}{\Adv}^*$ be a valid sequence of transactions that maximizes the loss $-\gain{\CmvC}{\WmvA \mid \cstD}{\TxTS}$.
    Since $\TxTS$ consists of transactions targeting contracts in $\CmvD \subseteq \cmvOfcst{\cstC}$ and since, by the well-formedness assumption, there are no internal calls from $\cstC$ to $\bar{\cstC}$, the contracts in $\bar{\cstC}$ are not affected by $\TxTS$.
    Hence, executing $\TxTS$ yields a transition of the form:
    \[
    \WmvA \mid \cstD 
    \; \xrightarrow{\TxTS} \;
    \WmvAi \mid \cstDi
    \tag*{\text{where $\bar{\cstC} \preceq \cstDi$}}
    \]
    As noted above, $\TxTS$ does not include any direct/indirect calls to $\cmvOfcst{\bar{\cstC}}$,  
    and so $\TxTS$ is also valid in $\WmvA \mid \cstC$. 
    Therefore, we also have some $\cstCi$ such that:
    \[
    \WmvA \mid \cstC 
    \; \xrightarrow{\TxTS} \; 
    \WmvAi \mid \cstCi
    \]
    To prove that the loss is constant, observe that:
    \begin{align*}
        \gain{\CmvC}{\WmvA \mid \cstC}{\TxTS}
        &= \wealth{\CmvC}{\WmvAi \mid \cstCi} - \wealth{\CmvC}{\WmvA \mid \cstC}
        \\
        &= \wealth{\CmvC}{\cstCi} - \wealth{\CmvC}{\cstC}
        \\
        &= \wealth{\CmvC}{\cstDi} - \wealth{\CmvC}{\bar{\cstC}} - \wealth{\CmvC}{\cstD} + \wealth{\CmvC}{\bar{\cstC}}
        \\
        &= \wealth{\CmvC}{\cstDi} - \wealth{\CmvC}{\cstD}
        \\
        &= \wealth{\CmvC}{\WmvAi \mid \cstDi} - \wealth{\CmvC}{\WmvA \mid \cstD}
        \\
        &= \gain{\CmvC}{\WmvA \mid \cstD}{\TxTS}
    \end{align*}
    This implies that:
    \begin{align*}
    \lmev{\CmvD}{\WmvA \mid \cstD}{\CmvC} \leq \lmev{\CmvD}{\WmvA \mid \cstC}{\CmvC}
    \end{align*}
    which gives our thesis.
    \qed

\end{proof}

\proofof{prop:qnonint:larger-state}
By~\Cref{def:qnonint}, we have two cases.

\medskip\noindent
If $\lmev{}{\sysS\mid\cstD}{\cmvOfcst{\cstD}} = 0$,
then  $\qnonint{\sysS}{\cstD} = 0$.
From~\Cref{lem:qnonint:basic}\ref{lem:qnonint:basic:3},
we have $\qnonint{\sysS \mid \cstC}{\cstD} \geq 0$.
This implies the thesis, $\qnonint{\sysS}{\cstD} \leq \qnonint{\sysS \mid \cstC}{\cstD}$.

\medskip\noindent
Otherwise, assume that $\lmev{}{\sysS\mid\cstD}{\cmvOfcst{\cstD}} > 0$. 
Then, by~\Cref{def:qnonint}:
\begin{align*}
        \qnonint{\sysS}{\cstD}
        &= 1 - 
        \frac
        {\lmev{\cmvOfcst{\cstD}}{\sysS \mid \cstD}{\cmvOfcst{\cstD}}}
        {\lmev{}{\sysS \mid \cstD}{\cmvOfcst{\cstD}}}
\end{align*}
Now, by~\Cref{lem:lmev:monotonicity} of~\Cref{lem:lmev}, we have that:
\begin{align*}
    0 < \lmev{}{\sysS\mid\cstD}{\cmvOfcst{\cstD}}
    & \leq \lmev{}{\sysS\mid\cstC\mid\cstD}{\cmvOfcst{\cstD}}
\end{align*}
Therefore, by~\Cref{def:qnonint}:
\begin{align*}
\qnonint{\sysS \mid \cstC}{\cstD}
& = 1 - 
\frac
{\lmev{\cmvOfcst{\cstD}}{\sysS \mid \cstC \mid \cstD}{\cmvOfcst{\cstD}}}
{\lmev{}{\sysS \mid \cstC \mid \cstD}{\cmvOfcst{\cstD}}}
\end{align*}

\noindent
From~\Cref{lem:lmev}, we have that:
\begin{align}
\lmev{\cmvOfcst{(\sysS \mid \cstD)}}{\sysS \mid \cstD}{\cmvOfcst{\cstD}} 
\nonumber
& \leq 
\lmev{\cmvOfcst{(\sysS\mid\cstD)}}{\sysS \mid \cstC \mid \cstD}{\cmvOfcst{\cstD}}
&& \text{by~\Cref{lem:lmev:monotonicity}} 
\\
& \leq 
\label{eq:qnonint:larger-state:1}
\lmev{\cmvOfcst{(\sysS\mid\cstC \mid \cstD)}}{\sysS \mid \cstC \mid \cstD}{\cmvOfcst{\cstD}}
&& \text{by~\Cref{lem:lmev:L-leq-H}} 
\end{align}

\noindent
We know from~\eqref{eq:qnonint:larger-state:1},
\[
    \lmev{}{\sysS\mid\cstD}{\cmvOfcst{\cstD}}
    \leq 
    \lmev{}{\sysS\mid\cstC\mid\cstD}{ \cmvOfcst{\cstD}}
\]
Taking the reciprocal on both sides gives us:
\[
    \frac{1}{\lmev{}{\sysS\mid\cstD}{\cmvOfcst{\cstD}}}
    \geq
    \frac{1}{\lmev{}{\sysS\mid\cstC\mid\cstD}{ \cmvOfcst{\cstD}}}
\]
By~\Cref{lem:lmev:prepending-state}, we have
\(
\lmev{\cmvOfcst{\cstD}}{{\sysS \mid \cstD}}{\cmvOfcst{\cstD}} = \lmev{\cmvOfcst{\cstD}}{{\sysS \mid \cstC \mid \cstD}}{\cmvOfcst{\cstD}}
\).
Then:
\[
    \frac{\lmev{\cmvOfcst{\cstD}}{\sysS\mid\cstD}{\cmvOfcst{\cstD}}}{\lmev{}{\sysS\mid\cstD}{\cmvOfcst{\cstD}}}
    \geq
    \frac{\lmev{\cmvOfcst{\cstD}}{\sysS\mid\cstC\mid\cstD}{\cmvOfcst{\cstD}}}{\lmev{}{\sysS\mid\cstC\mid\cstD}{ \cmvOfcst{\cstD}}}
\]
which finally gives us:
\[
    1 - \frac{\lmev{\cmvOfcst{\cstD}}{\sysS\mid\cstD}{\cmvOfcst{\cstD}}}{\lmev{}{\sysS\mid\cstD}{\cmvOfcst{\cstD}}}
    \leq
    1 - \frac{\lmev{\cmvOfcst{\cstD}}{\sysS\mid\cstC\mid\cstD}{\cmvOfcst{\cstD}}}{\lmev{}{\sysS\mid\cstC\mid\cstD}{ \cmvOfcst{\cstD}}}
\]
which gives our thesis, \ie $\qnonint{\sysS}{\cstD} \leq \qnonint{\sysS \mid \cstC}{\cstD}$.
\qed

\proofof{prop:qnonint:adv wallets}
By~\Cref{lem:lmev-wallet} we have that, for all $\CmvC,\CmvD$ and for all $\cstCi$:
\[
    \dom{\WmvA[\Adv]} = \Adv \implies \lmev{\CmvD}{\WmvA[\Adv] \mid \WmvA \mid \cstCi}{\CmvC} = \lmev{\CmvD}{\WmvA[\Adv] \mid \cstCi}{\CmvC}
\]
In particular, by choosing $\cstCi = \cstC \mid \cstD$ 
and $\CmvC = \cmvOfcst{\cstD}$,
this implies that:
\begin{align*}
    & \lmev{}{\WmvA[\Adv]\mid\WmvA\mid\cstC\mid\cstD}{\cmvOfcst{\cstD}} = \lmev{}{\WmvA[\Adv]\mid\cstC\mid\cstD}{\cmvOfcst{\cstD}} 
    \\
    & \lmev{\cmvOfcst{\cstD}}{\WmvA[\Adv]\mid\WmvA\mid\cstC\mid\cstD}{\cmvOfcst{\cstD}} = \lmev{\cmvOfcst{\cstD}}{\WmvA[\Adv]\mid\cstC\mid\cstD}{\cmvOfcst{\cstD}}
\end{align*}
which gives us our thesis, \ie
$
    \qnonint{\WmvA[\Adv] \mid \WmvA \mid \cstC}{\cstD} = \qnonint{\WmvA[\Adv] \mid \cstC}{\cstD}
$
\qed

% \medskip
% We now state the auxiliary notion of \emph{stripping}, which is subsequently used in~\Cref{th:lmev:contract-stripping,th:qnonint:preserving-interference}.
% Note that our notion is alternative to the one in~\cite{BMZ24fc}.

% \begin{definition}[Stripping]
% Given two sets of contracts $\CmvD, \CmvC$ in $\sysS$, we define the stripping of $\CmvD$ \wrt $\CmvC$ (denoted by $\strip{\CmvD}{\CmvC}$) as $\CmvD \cap \deps{\CmvC}$.
% \end{definition}

% \begin{figure}[t]
% \begin{lstlisting}[
% ,language=txscript
% ,morekeywords={C,C',withdraw},classoffset=4
% ,morekeywords={a,b,A,Oracle},keywordstyle=\addrColor,classoffset=5
% ,morekeywords={t,t1,t2,T,tin,tout},keywordstyle=\tokColor,classoffset=6
% ,caption={Contracts to demonstrate sender-agnosticity.}
% ,label={lst:sender-agnosticity}
% ]
% contract (*$\contract{C}$*) {
%   withdraw() { require(sender==(*$\contract{D}$*)); transfer(sender,n:T); }
% }
% contract (*$\contract{D}$*) {
%   withdraw() { (*$\contract{C}$*).withdraw(); transfer(sender,n:T); }
% }
% \end{lstlisting}
% \end{figure}

\medskip
\Cref{th:lmev:contract-stripping} gives sufficient conditions under which we can strip $\CmvD$ from all the non-dependencies of $\CmvC$ while preserving $\lmev{\CmvD}{\sysS}{\CmvC}$.
Condition~\ref{th:lmev:contract-stripping:1} is that contract functions are sender-agnostic, \ie they are not aware of the identity of the $\sender$, being only able to use it as a recipient of token transfers.
Condition~\ref{th:lmev:contract-stripping:2} ensures that $\CmvD$ contains enough contracts to reproduce attacks in the stripped state.
Condition~\ref{th:lmev:contract-stripping:3} requires that the dependencies and the non-dependencies of $\CmvC$ in $\CmvD$ are token independent in $\sysS$.
In other words, there are no token dependencies between $\strip{\CmvD}{\CmvC}$ and $\CmvD \setminus \deps{\CmvC}$, which could have potentially be exploited by non-wealthy adversaries.

% To illustrate sender-agnosticism, consider the contracts in~\Cref{lst:sender-agnosticity}.
% Note that contract $\contract{C}$ violates sender-agnosticism, since $\sender$ is used to enable token transfers only to contract $\contract{D}$.

\begin{lemma}
\label{th:lmev:contract-stripping}
The equality:
\[
\lmev{\CmvD}{\sysS}{\CmvC} = \lmev{\strip{\CmvD}{\CmvC}}{\sysS}{\CmvC}
\]
holds if all the following conditions, 
where $\CmvCi = \deps{\CmvC} \cap \deps{\CmvD \setminus \deps{\CmvC}}$,
are satisfied:
\begin{enumerate}

\item \label{th:lmev:contract-stripping:1}
the contracts in $\CmvCi$ are sender-agnostic;
    
\item \label{th:lmev:contract-stripping:2}
$\CmvCi \subseteq \CmvD$;

\item \label{th:lmev:contract-stripping:3}
% $\strip{\CmvD}{\CmvC}$ and $\CmvD \setminus \deps{\CmvC}$ are token independent in $\sysS$.
$\strip{\deps{\CmvD}}{\CmvC}$ and $\deps{\CmvD} \setminus \deps{\CmvC}$ are token independent in $\sysS$.

\end{enumerate}
\end{lemma}
\begin{proof}
First, note that the inequality $\lmev{\strip{\CmvD}{\CmvC}}{\sysS}{\CmvC} \leq \lmev{\CmvD}{\sysS}{\CmvC}$ follows from \Cref{lem:lmev:L-leq-H} of \Cref{lem:lmev}, so we just need to show that:
\[
\lmev{\CmvD}{\sysS}{\CmvC} \leq \lmev{\strip{\CmvD}{\CmvC}}{\sysS}{\CmvC}
\]
To do so, let $\TxTS \in \mall{\CmvD}{\Adv}^*$ be a sequence of transactions that maximizes the loss of $\CmvC$ when executed in state $\sysS$. 
We show that there exists $\TxYS \in \mall{\strip{\CmvD}{\CmvC}}{\Adv}^*$ that causes a loss to $\CmvC$ equal to the one caused by $\TxTS$, \ie:
\begin{equation}
\label{eq:lmev:contract-stripping:Y}
\TxYS \in \mall{\strip{\CmvD}{\CmvC}}{\Adv}^*
\qquad
\gain{\CmvC}{\sysS}{\TxYS} = \gain{\CmvC}{\sysS}{\TxTS}
\end{equation}
\Wlog we assume that all the transactions in $\TxTS$ are valid:
indeed, invalid transactions in $\TxTS$ are reverted, so they can be removed without affecting the loss.

\medskip\noindent
Note that each transaction $\txT[i] = \pmvM[i]:\contract[i,1]{C}.\txcode{f_{i,1}}(\code{args_{i,1}})$ in $\TxTS$ can trigger a sequence of \emph{internal} contract-to-contract function calls:
\[
\contract[i,1]{C}:\contract[i,2]{C}.\txcode{f_{i,2}}(\code{args_{i,2}})
\;\;
\contract[i,2]{C}:\contract[i,3]{C}.\txcode{f_{i,3}}(\code{args_{i,3}})
\;\; \cdots \;\;
\contract[i,k-1]{C}:\contract[i,k]{C}.\txcode{f_{i,k}}(\code{args_{i,k}})
\]
Let $\vec{x}$ be the sequence of all function calls (either external or internal) that are performed upon the execution of $\TxTS$ in state $\sysS$.
To construct $\TxYS$, we start by considering
% Note that $\vec{x}$ includes both \emph{external} calls sent directly from a transaction in $\TxTS$ to a contract in $\CmvD$ as well as \emph{internal} calls that are performed from another called function. 
the subsequence $\vec{y}$ of $\vec{x}$ containing all and only the calls of the form:
\begin{enumerate}

\item[(a)] $\pmvM[i]:\contract[i,1]{C}.\txcode{f_{i,1}}(\code{args_{i,1}})$ where $\contract[i,1]{C} \in \deps{\CmvC}$, or

\item[(b)] $\contract[i,j-1]{C}:\contract[i,j]{C}.\txcode{f_{i,j}}(\code{args_{i,j}})$, where $\contract[i,j-1]{C} \not\in \deps{\CmvC}$ and  $\contract[i,j]{C} \in \deps{\CmvC}$.

% \item[(a)] a transaction of $\TxTS$ directly calling a function of a contract in $\deps{\CmvC}$, or

% \item[(b)] a function of a contract \emph{not} in $\deps{\CmvC}$ calling a function of a contract in $\deps{\CmvC}$.

\end{enumerate}

\noindent
\textbf{Claim (1).} If $\contract[i,j-1]{C}:\contract[i,j]{C}.\txcode{f_{i,j}}(\code{args_{i,j}}) \in \vec{y}$, 
% due to condition (b), 
then $\contract[i,j]{C} \in \CmvCi$.
% $\CmvCi = \deps{\CmvC} \cap \deps{\CmvD \setminus \deps{\CmvC}}$.

\medskip\noindent
\emph{Proof of Claim (1).} 
By hypothesis, $\contract[i,j]{C} \in \deps{\CmvC}$. 
% Moreover, we know that it has been called internally from a contract not in $\deps{\CmvC}$, and that this call has originated from a transaction $\txT[i]$ in $\TxTS$.
Let $\txT[i] \in \TxTS$ be the transaction that originated the call.
Since $\txT[i] \in \mall{\CmvD}{\Adv}$, then  
$\contract[i,1]{C} \in \CmvD$.
Since $\deps{\CmvC}$ is closed downward and 
$\contract[i,j-1]{C} \not\in \deps{\CmvC}$,
then $\contract[i,1]{C} \not\in \deps{\CmvC}$. 
So, $\contract[i,1]{C} \in \CmvD \setminus \deps{\CmvC}$, and therefore $\contract[i,j]{C} \in \deps{\CmvD \setminus \deps{\CmvC}}$.
This completes the proof of Claim (1).

\medskip
To describe the construction of $\TxYS$, let the meta-variables $\addr[i]{a}$ range over user and contract addresses, so to rewrite the sequence $\vec{y}$ as follows:
\begin{align*}
& \addr[1]{a}:\contract[1]{C}.\txcode{f_{1}}(\code{args_{1}})
\;\;
\addr[2]{a}:\contract[2]{C}.\txcode{f_{2}}(\code{args_{2}})
\;\; \cdots \;\;
\addr[n]{a}:\contract[n]{C}.\txcode{f_{n}}(\code{args_{n}})
\cdots
% \\
% & \pmvM[n]:\contract[n,1]{C}.\txcode{f_{n,1}}(\code{args_{n,1}})
% \;\;
% \contract[n,1]{C}:\contract[n,2]{C}.\txcode{f_{n,2}}(\code{args_{n,2}})
% \;\; \cdots \;\;
% \contract[n,kn-1]{C}:\contract[n,k]{C}.\txcode{f_{n,k}}(\code{args_{n,kn}})
\end{align*}
We translate $\vec{y}$ into the sequence of \emph{transactions} $\TxYS$
by preserving the senders $\addr[i]{a}$ that are user accounts (\ie, $\addr[i]{a} = \pmvM[i]$),
and by replacing the $\addr[i]{a}$ that are contract accounts into the user account that originated the corresponding call.
Namely, if $\addr[i]{a} = \contract[i,j-1]{C}$ is a contract account corresponding to the following call in $\vec{y}$:
\[
\contract[i,j-1]{C}:\contract[i,j]{C}.\txcode{f_{i,j}}(\code{args_{i,j}})
\]
then the sender of the $i$-th transaction in $\TxYS$ is $\pmvM[i]$, \ie the originator of the call.
% We now let $\TxYS$ be the sequence of transactions that directly perform the calls in $\vec{y}$, in the same order, with the same arguments and origin. 
%
Note that each transaction $\txY[i]$ in $\TxYS$ can be funded by the adversary:
\begin{itemize}

\item if $\addr[i]{a} = \pmvM[i]$, then the fact that the corresponding transaction $\txT[i]$ in $\TxTS$ was valid implies that $\pmvM[i]$ has the tokens needed to fund the call; 

\item if $\addr[i]{a} = \contract[i,j-1]{C}$, then there is no token transfer from $\contract[i,j-1]{C}$ to $\contract[i,j]{C}$, and so $\txY[i]$ does not need to be funded.
This is because:
\begin{itemize}

\item $\contract[i,j-1]{C} \in \deps{\CmvD} \setminus \deps{\CmvC}$: 
indeed, $\contract[i,j-1]{C} \in \deps{\CmvD}$ since $\txT[i] \in \mall{\CmvD}{\Adv}$, and $\contract[i,j-1]{C} \not\in \deps{\CmvC}$ by definition of case (b);

\item $\contract[i,j]{C} \in \deps{\CmvD} \cap \deps{\CmvC}$: 
indeed, $\contract[i,j]{C} \in \deps{\CmvD}$ since $\txT[i] \in \mall{\CmvD}{\Adv}$, and $\contract[i]{C} \in \deps{\CmvC}$ by definition of case (b);

\item $\strip{\deps{\CmvD}}{\CmvC}$ and $\deps{\CmvD} \setminus \deps{\CmvC}$ are token independent in $\sysS$ by assumption~\eqref{th:lmev:contract-stripping:3}.

\end{itemize}

\end{itemize}
% Furthermore, if any of these internal calls transferred some tokens from $\contract[i,j-1]{C}$ to $\contract[i,j]{C}$, then the corresponding sender $\pmvM[i]$ of $\txY[i]$ provides the same amount of tokens.
% This is possible since the adversary has enough tokens to fund all the calls in $\TxYS$: indeed, $\Adv$ has enough funds to execute $\TxTS$, and any token that she gains from the discarded transactions 
% (\ie, those with callees in $\CmvD \setminus \deps{\CmvC}$) 
% cannot be used by calls in $\TxYS$, due to the token independence between $\strip{\CmvD}{\CmvC}$ and $\CmvD \setminus \deps{\CmvC}$. %
% (assumption~\ref{th:lmev:contract-stripping:3}).

\noindent
\textbf{Claim (2).} $\TxYS \in \mall{\strip{\CmvD}{\CmvC}}{\Adv}^*$
    
\medskip\noindent
\emph{Proof of Claim (2).} 
Consider a transaction $\txY[i]$ in $\TxYS$.
We have two cases, depending on whether $\txY[i]$ is due to conditions (a) or (b):
\begin{enumerate}

\item[(a)] in this case, $\txY[i]$ corresponds to some
$\txT[i] = \pmvM[i]:\contract[i,1]{C}.\txcode{f_{i,1}}(\code{args_{i,1}})$ in $\TxTS$ where $\contract[i,1]{C} \in \deps{\CmvC}$.
% transaction in $\TxTS$ targeting a contract in $\deps{\CmvC}$. 
Since $\txT[i] \in \mall{\CmvD}{\Adv}$, then $\txY[i] \in \mall{\strip{\CmvD}{\CmvC}}{\Adv}$. 

\item[(b)] by Claim (1), the callee of $\txY[i]$ is in $\CmvCi = \deps{\CmvC} \cap \deps{\CmvD \setminus \deps{\CmvC}}$, which is included in $\CmvD$ by assumption~\ref{th:lmev:contract-stripping:2}. Note that $\Adv$ is able to craft the actual arguments of that call by simulating the execution of $\TxTS$. This implies that  $\txY[i] \in \mall{\strip{\CmvD}{\CmvC}}{\Adv}$.
This completes the proof of Claim (2).

\end{enumerate}

We now show that $\TxYS$ and $\TxTS$ modify the state of contracts in $\CmvC$ in exactly the same way. 
Note that the transactions $\txY[i]$ that are in $\TxYS$ due to condition (b) have callee in $\CmvCi$ by Claim (1), and so their functions
are \emph{sender-agnostic} by assumption~\ref{th:lmev:contract-stripping:1}. 
% \contract[i,j-1]{C}:\contract[i,j]{C}.\txcode{f_{i,j}}(\code{args_{i,j}})
So, the fact that in the execution of $\txY[i]$ they are called directly from a user address, while in the execution of $\txT[i]$ they are called from a contract address, does not affect the execution of these calls. 
% relatively to~$\CmvC$.
Note that a call $\contract[i,j-1]{C}:\contract[i,j]{C}.\txcode{f_{i,j}}(\code{args_{i,j}})$ in $\txT[i]$ could send tokens to the sender $\contract[i,j-1]{C}$, thus affecting its gain, while the corresponding call $\pmvM[i]:\contract[i,j]{C}.\txcode{f_{i,j}}(\code{args_{i,j}})$ would send these tokens to $\pmvM[i]$.
This difference however do not affect the gains and losses of $\CmvC$, since $\contract[i,j-1]{C}$ is not in $\deps{\CmvC}$ by condition (b).

Note that the sequence $\vec{h}$ of calls performed upon the execution of $\TxYS$ contains $\vec{y}$ but does not coincide with it, since it also includes all the internal calls that are performed by functions in $\vec{y}$.
\emnote{I think the previous sentence is not true. This is because I do not think $\vec{h}$ contains $\vec{y}$: $\vec{y}$ (subsequence of $\vec{x}$) contains calls from a non-dependencies of $\CmvC$ to $\deps{\CmvC}$, whereas in $\vec{h}$ we replace the sender (non-dependency of $\CmvC$) with the adversarial account that originated the corresponding call. So $\vec{h}$ does not contain $y$. But the following sentence seems okay. Please let me know if my query is clear or whether I should re-write it in formal language.}
In fact, $\vec{h}$ is the subsequence of $\vec{x}$ that contains every call to functions of contracts in $\deps{\CmvC}$. For this reason, both $\vec{x}$ and $\vec{h}$ modify the state of contracts % $\strip{\CmvD}{\CmvC}$ 
$\deps{\CmvC}$ in the same way --- and, in particular, they cause exactly the same losses to the contracts in $\CmvC$.
% \barterror{but here we must preserve the loss of contracts $\CmvC$! If we additionally assume that $\CmvC \subseteq \CmvD$, then the state of $\CmvD \cap \deps{\CmvC} \supseteq \CmvC \cap \deps{\CmvC} = \CmvC$ is preserved! Is this enough? Is it true that preserving state implies transferring the same tokens?}
This implies that $\TxYS$ is valid in $\sysS$ and that $\gain{\CmvC}{\sysS}{\TxYS} = \gain{\CmvC}{\sysS}{\TxTS}$. 
% since contracts in $\CmvC$ are not affected by the absence of contracts outside of $\strip{\CmvD}{\CmvC}$.
Since we have proved~\eqref{eq:lmev:contract-stripping:Y} for all possible $\TxTS$, we obtain the thesis. 
\qed
\end{proof}

\begin{example}
\label{ex:lmev:contract-stripping}
To illustrate~\Cref{th:lmev:contract-stripping}, consider the contracts: % in~\Cref{lst:cex:lmev:contract-stripping}.

\begin{lstlisting}[
,language=txscript
,basicstyle=\fontseries{m}\normalsize\ttfamily\lst@ifdisplaystyle\footnotesize\fi,
,morekeywords={f,g,f0,f1,g0,g1},classoffset=4
,morekeywords={a,A,M},keywordstyle=\pmvColor
,classoffset=5,morekeywords={t,T,T0,T1,T2,ETH}
,keywordstyle=\tokColor,classoffset=6
,morekeywords={C0,C1,C2,C3},keywordstyle=\cmvColor
%,caption={Illustration of~\Cref{th:lmev:contract-stripping}}
%,label={lst:cex:lmev:contract-stripping}
]
contract C0 { f(a pays 1:T) { transfer(M,2:T) } }
contract C1 { f() { C0.f(C1 pays 1:T); } }
contract C2 { f() { require sender==C3; C1.f(); } }
contract C3 { f() { C2.f(); C1.f(); } }
\end{lstlisting}

\noindent
Let 
$\Adv = \setenum{\pmvM}$, 
$\CmvD= \setenum{\contract{C0},\contract{C1},\contract{C3}}$,
$\CmvC = \setenum{\contract{C0,C1}}$, 
and let:
\[
\sysS = 
\walpmv{\pmvM}{0:\tokT} \mid
\walpmv{\contract{C0}}{\waltok{2}{\tokT}} \mid 
\walpmv{\contract{C1}}{\waltok{2}{\tokT}} \mid 
\walpmv{\contract{C2}}{\waltok{0}{\tokT}} \mid
\walpmv{\contract{C3}}{\waltok{0}{\tokT}}
\]
Let $\TxTS \in \mall{\CmvD}{\pmvM}$ be the following sequence of transactions:
\[
\TxTS = \pmvM:\contract{C3}.\txcode{f()}
\]
By executing $\TxTS$ in $\sysS$, we have that:
\begin{align*}
    \sysS
    & \xrightarrow{\pmvM:\contract{C3}.\txcode{f()}}
    \walpmv{\pmvM}{4:\tokT} \mid
    \walpmv{\contract{C0}}{\waltok{0}{\tokT}} \mid 
    \walpmv{\contract{C1}}{\waltok{0}{\tokT}} \mid 
    \walpmv{\contract{C2}}{\waltok{0}{\tokT}} \mid
    \walpmv{\contract{C3}}{\waltok{0}{\tokT}}
\end{align*}
Since there are no tokens left in $\CmvC$, $\TxTS$ clearly maximises the loss of $\CmvC$, hence:
\[
\lmev{\CmvD}{\sysS}{\CmvC} = 4 \cdot \price{\tokT}
\]
We first check that the conditions of~\Cref{th:lmev:contract-stripping}
are satisfied.
Let:
\begin{align*}
    \CmvCi 
    & = \deps{\CmvC} \cap \deps{\CmvD \setminus \deps{\CmvC}}
    = \setenum{\contract{C0},\contract{C1}} \cap
    \deps{\setenum{\contract{C0},\contract{C1},\contract{C3}} \setminus \setenum{\contract{C0},\contract{C1}}} 
    \\
    & = \setenum{\contract{C0},\contract{C1}} \cap \deps{\setenum{\contract{C3}}}
     = \setenum{\contract{C0},\contract{C1}} \cap \setenum{\contract{C0},\contract{C1},\contract{C2},\contract{C3}}
     \\
     & = \setenum{\contract{C0},\contract{C1}}
\end{align*}
The conditions of~\Cref{th:lmev:contract-stripping} are then satisfied, since:
\begin{itemize}

\item[\eqref{th:lmev:contract-stripping:1}]
the contracts $\contract{C0},\contract{C1} \in \CmvCi$ is sender-agnostic.
Note that $\contract{C2}$ is not sender-agnostic, but this does not violate assumption~\ref{th:lmev:contract-stripping:1} since sender-agnosticism is only required on $\CmvCi$;

\item[\eqref{th:lmev:contract-stripping:2}]
$\CmvCi = \setenum{\contract{C0},\contract{C1}} \subseteq \CmvD$.
Note that this inclusion is stricter than necessary: indeed, in this example, choosing $\CmvD = \setenum{\contract{C1},\contract{C3}}$ would have violated assumption~\ref{th:lmev:contract-stripping:2}, but it would have still preserved the MEV (see $\TxYS$ below).

\item[\eqref{th:lmev:contract-stripping:3}]
token independence of the parts of $\sysS$ related to contracts:
\begin{align*}
\deps{\CmvD} \cap \deps{\CmvC} 
& = 
\setenum{\contract{C0},\contract{C1},\contract{C2},\contract{C3}} \cap \setenum{\contract{C0},\contract{C1}} 
= \setenum{\contract{C0},\contract{C1}}
\\
\deps{\CmvD} \setminus \deps{\CmvC} 
& = \setenum{\contract{C0},\contract{C1},\contract{C2},\contract{C3}} \setminus \setenum{\contract{C0},\contract{C1}} 
= \setenum{\contract{C2},\contract{C3}} 
\end{align*}
Note instead that token independence is not required between $\contract{C0}$ and $\contract{C1}$: actually, these two contracts are token dependent, since $\contract{C1}$ sends $1:\tokT$ along with the internal call to $\contract{C0}$.
\end{itemize}

\noindent
We now construct the sequence of transactions $\TxYS$ following the proof of \Cref{th:lmev:contract-stripping}.
The sequence $\vec{x}$ of calls induced by $\TxTS$, the subsequence $\vec{y}$ obtained by filtering $\vec{x}$, and the sequence of transactions $\TxYS$ are the following:
\[
\begin{array}{lllllll}
\vec{x} \; =
\quad
& \pmvM:\contract{C3}.\txcode{f()}
\quad
& \contract{C3}:\contract{C2}.\txcode{f()}
\quad
& \contract{C2}:\contract{C1}.\txcode{f()}
\quad
& \contract{C1}:\contract{C0}.\txcode{f()}
\quad
& \contract{C3}:\contract{C1}.\txcode{f()}
\quad
& \contract{C1}:\contract{C0}.\txcode{f()}
\\
\vec{y} \; =
\quad
& % \pmvM:\contract{C3}.\txcode{f()}
% \quad
& % \contract{C3}:\contract{C2}.\txcode{f()}
% \quad
& \contract{C2}:\contract{C1}.\txcode{f()}
\quad
& 
% \contract{C1}:\contract{C0}.\txcode{f()}
\quad
& \contract{C3}:\contract{C1}.\txcode{f()}
\quad
& 
% \contract{C1}:\contract{C0}.\txcode{f()}
\\
\TxYS  =
& % \pmvM:\contract{C3}.\txcode{f()}
% \quad
& % \contract{C3}:\contract{C2}.\txcode{f()}
% \quad
& \pmvM:\contract{C1}.\txcode{f()}
\quad
&
& \pmvM:\contract{C1}.\txcode{f()}
&
\end{array}
\]
Note that $\TxYS \in \mall{\CmvD \cap \deps{\CmvC}}{\pmvM} = \mall{\setenum{\contract{C0},\contract{C1}}}{\pmvM}$. 
By executing $\TxYS$ in $\sysS$, we have that:
\begin{align*}
    \sysS
    & \xrightarrow{\pmvM:\contract{C1}.\txcode{f()}}
    \walpmv{\pmvM}{2:\tokT} \mid
    \walpmv{\contract{C0}}{\waltok{1}{\tokT}} \mid 
    \walpmv{\contract{C1}}{\waltok{1}{\tokT}} \mid 
    \walpmv{\contract{C2}}{\waltok{0}{\tokT}} \mid
    \walpmv{\contract{C3}}{\waltok{0}{\tokT}}
    \\
    & \xrightarrow{\pmvM:\contract{C1}.\txcode{f()}}
    \walpmv{\pmvM}{4:\tokT} \mid
    \walpmv{\contract{C0}}{\waltok{0}{\tokT}} \mid 
    \walpmv{\contract{C1}}{\waltok{0}{\tokT}} \mid 
    \walpmv{\contract{C2}}{\waltok{0}{\tokT}} \mid
    \walpmv{\contract{C3}}{\waltok{0}{\tokT}}
\end{align*}
Hence, we have that:
\[
\lmev{\CmvD \cap \deps{\CmvC}}{\sysS}{\CmvC} = 4 \cdot \price{\tokT}
\]
which confirms the preservation of MEV stated by~\Cref{th:lmev:contract-stripping}.
\hfill\qedex
\end{example}

% \begin{figure}[t]
% \begin{lstlisting}[
% ,language=txscript
% ,morekeywords={f,g},classoffset=4
% ,morekeywords={a,A,Oracle},keywordstyle=\pmvColor
% ,classoffset=5,morekeywords={t,T,T0,T1,T2,ETH}
% ,keywordstyle=\tokColor,classoffset=6
% ,morekeywords={C0,C1,D},keywordstyle=\cmvColor
% ,caption={Illustrating the token independence assumption  in~\Cref{th:lmev:contract-stripping}.}
% ,label={lst:ex:lmev:contract-stripping:token-independence}]
%     contract C0 { f() { D.g(); sender!1:ETH } }
%     contract C1 { f() { D.g(); sender!1:ETH } }
%     contract D  { g() { sender!1:T } }
% \end{lstlisting}
% \end{figure}

\begin{example}
\label{ex:lmev:contract-stripping:token-independence}
To illustrate the need of the token independence assumption in \Cref{th:lmev:contract-stripping}, consider the contracts: 
% \Cref{lst:ex:lmev:contract-stripping:token-independence}.

\begin{lstlisting}[
,language=txscript
,basicstyle=\fontseries{m}\normalsize\ttfamily\lst@ifdisplaystyle\footnotesize\fi,
,morekeywords={f,g,f0,f1,g0,g1,receive},classoffset=4
,morekeywords={a,A,M},keywordstyle=\pmvColor
,classoffset=5,morekeywords={t,T,T0,T1,T2,ETH}
,keywordstyle=\tokColor,classoffset=6
,morekeywords={C0,C1,C2,C3},keywordstyle=\cmvColor
%,caption={Illustration of~\Cref{th:lmev:contract-stripping}}
%,label={lst:cex:lmev:contract-stripping}
]
contract C0 { f(a pays 2:T) { transfer(M,4:T) } }
contract C1 { 
    f() { C0.f(C1 pay 2:T); }
    receive(a pays n:T) { }
}
contract C2 { 
  f() { C1.receive(C2 pays 1:T); }
     // transfer(C1,1:T) is forbidden in our model
     // transfer recipients must be user accounts
}
contract C3 { f() { C2.f(); C1.f(); } }
\end{lstlisting}

\noindent
Let 
$\Adv = \setenum{\pmvM}$, 
$\CmvC = \setenum{\contract{C0},\contract{C1}}$, 
$\CmvD= \setenum{\contract{C0},\contract{C1},\contract{C3}}$, and let:
\[
\sysS = 
\walpmv{\pmvM}{0:\tokT} \mid
\walpmv{\contract{C0}}{\waltok{2}{\tokT}} \mid 
\walpmv{\contract{C1}}{\waltok{1}{\tokT}} \mid 
\walpmv{\contract{C2}}{\waltok{1}{\tokT}} \mid
\walpmv{\contract{C3}}{\waltok{0}{\tokT}}
\]
Note that $\pmvM$ has no tokens in $\sysS$, to the only way to extract MEV is to pass through $\contract{C3}$.
Let $\TxTS \in \mall{\CmvD}{\pmvM}$ be the following sequence of transactions:
\[
\TxTS = \pmvM:\contract{C3}.\txcode{f()}
\]
By executing $\TxTS$ in $\sysS$, we have that:
\begin{align*}
    \sysS
    & \xrightarrow{\pmvM:\contract{C3}.\txcode{f()}}
    \walpmv{\pmvM}{4:\tokT} \mid
    \walpmv{\contract{C0}}{\waltok{0}{\tokT}} \mid 
    \walpmv{\contract{C1}}{\waltok{0}{\tokT}} \mid 
    \walpmv{\contract{C2}}{\waltok{0}{\tokT}} \mid
    \walpmv{\contract{C3}}{\waltok{0}{\tokT}}
\end{align*}
Since there are no tokens left in $\CmvC$, $\TxTS$ clearly maximises the loss of $\CmvC$.
Since $\CmvC$ contained $3:\tokT$ in $\sysS$, then:
\[
\lmev{\CmvD}{\sysS}{\CmvC} = 3 \cdot \price{\tokT}
\]
Similarly to~\Cref{ex:lmev:contract-stripping}, we have that $\CmvCi = \setenum{\contract{C0},\contract{C1}}$, which satisfies conditions \eqref{th:lmev:contract-stripping:1} and \eqref{th:lmev:contract-stripping:2}. 
For condition \eqref{th:lmev:contract-stripping:3} (token independence), we have that:
\begin{align*}
\deps{\CmvD} \cap \deps{\CmvC} 
& = 
\setenum{\contract{C0},\contract{C1},\contract{C2},\contract{C3}} \cap \setenum{\contract{C0},\contract{C1}} 
= \setenum{\contract{C0},\contract{C1}}
\\
\deps{\CmvD} \setminus \deps{\CmvC} 
& = \setenum{\contract{C0},\contract{C1},\contract{C2},\contract{C3}} \setminus \setenum{\contract{C0},\contract{C1}} 
= \setenum{\contract{C2},\contract{C3}} 
\end{align*}
Now, token independence between 
$\setenum{\contract{C0},\contract{C1}}$
and
$\setenum{\contract{C2},\contract{C3}}$
is not satisfied, since $\contract{C2}$ sends $1:\tokT$ to $\contract{C1}$.
More formally, we have that:
\[
\intok{\sysS}{\cstC[\contract{C0},\contract{C1}]} = \setenum{\tokT}
\quad
\outtok{\sysS}{\cstC[\contract{C2},\contract{C3}]} = \setenum{\tokT}
\quad
\intok{\sysS}{\cstC[\contract{C2},\contract{C3}]} = \emptyset
\quad
\outtok{\sysS}{\cstC[\contract{C0},\contract{C1}]} = \setenum{\tokT}
\]
By~\Cref{def:token-independence}, since:
\[
\intok{\sysS}{\cstC[\contract{C0},\contract{C1}]} \cap  
\outtok{\sysS}{\cstC[\contract{C2},\contract{C3}]} 
= \setenum{\tokT} \cap \setenum{\tokT}
\neq \emptyset
\]
then, $\cstC[\contract{C0},\contract{C1}]$ and $\cstC[\contract{C2},\contract{C3}]$
are \emph{not} token independent.

\medskip\noindent
Since the conditions of~\Cref{th:lmev:contract-stripping} are \emph{not} satisfied, we are not guaranteed to have the preservation of MEV:
\[
\lmev{\CmvD \cap \deps{\CmvC}}{\sysS}{\CmvC}
\; \stackrel{?}{=} \;
\lmev{\CmvD}{\sysS}{\CmvC} = 3 \cdot \price{\tokT}
\]
Indeed, the maximum loss that $\Adv$ can inflict to $\CmvC$ using $\mall{\CmvD \cap \deps{\CmvC}}{\Adv} = \mall{\setenum{\contract{C0},\contract{C1}}}{\Adv}$ is zero.
This is because:
\begin{itemize}
\item calling $\contract{C0}$ fails, since $\contract{C0}$ has not the required $4:\tokT$ to transfer;
\item calling $\contract{C1}$ fails, since $\contract{C1}$ does not have the $2:\tokT$ required to call $\contract{C0}$. 
\end{itemize} 
Note also that requiring the milder condition that
$\CmvD \cap \deps{\CmvC}$ and $\CmvD \setminus \deps{\CmvC}$ are token independent would not be enough to guarantee MEV preservation.
In our example, we would have:
\begin{align*}
\CmvD \cap \deps{\CmvC} 
& = 
\setenum{\contract{C0},\contract{C1},\contract{C3}} \cap \setenum{\contract{C0},\contract{C1}} 
= \setenum{\contract{C0},\contract{C1}}
\\
\CmvD \setminus \deps{\CmvC} 
& = \setenum{\contract{C0},\contract{C1},\contract{C3}} \setminus \setenum{\contract{C0},\contract{C1}} 
= \setenum{\contract{C3}} 
\end{align*}
where $\cstC[\contract{C0},\contract{C1}]$ and $\cstC[\contract{C3}]$
are token independent.
\hfill\qedex
\end{example}

\begin{example}
\label{ex:lmev:contract-stripping:sender-agnostic}
To illustrate the need of the sender-agnosticism assumption in \Cref{th:lmev:contract-stripping}, consider the contracts: 

\begin{lstlisting}[
,language=txscript
,basicstyle=\fontseries{m}\normalsize\ttfamily\lst@ifdisplaystyle\footnotesize\fi,
,morekeywords={f,g,f0,f1,g0,g1},classoffset=4
,morekeywords={a,A,M},keywordstyle=\pmvColor
,classoffset=5,morekeywords={t,T,T0,T1,T2,ETH}
,keywordstyle=\tokColor,classoffset=6
,morekeywords={C0,C1,C2,C3},keywordstyle=\cmvColor
%,caption={Illustration of~\Cref{th:lmev:contract-stripping}}
%,label={lst:cex:lmev:contract-stripping}
]
contract C0 { f() { require sender==C1; transfer(M,1:T) }}
contract C1 { f() { C0.f(); } }
\end{lstlisting}

\noindent
Let 
$\Adv = \setenum{\pmvM}$, 
$\CmvC = \setenum{\contract{C0}}$, 
$\CmvD= \setenum{\contract{C0},\contract{C1}}$, and let:
\[
\sysS = 
\walpmv{\pmvM}{0:\tokT} \mid
\walpmv{\contract{C0}}{\waltok{1}{\tokT}} \mid 
\walpmv{\contract{C1}}{\waltok{0}{\tokT}}
\]
Let $\TxTS \in \mall{\CmvD}{\pmvM}$ be the following sequence of transactions:
\[
\TxTS = \pmvM:\contract{C1}.\txcode{f()}
\]
By executing $\TxTS$ in $\sysS$, we have that:
\begin{align*}
    \sysS
    & \xrightarrow{\pmvM:\contract{C1}.\txcode{f()}}
    \walpmv{\pmvM}{1:\tokT} \mid
    \walpmv{\contract{C0}}{\waltok{0}{\tokT}} \mid 
    \walpmv{\contract{C1}}{\waltok{0}{\tokT}}
\end{align*}
Since there are no tokens left in $\CmvC$, $\TxTS$ clearly maximises the loss of $\CmvC$, hence:
\[
\lmev{\CmvD}{\sysS}{\CmvC} = 1 \cdot \price{\tokT}
\]
We have that $\CmvCi = \setenum{\contract{C0}}$, which satisfies condition \eqref{th:lmev:contract-stripping:2}. 
Note that condition \eqref{th:lmev:contract-stripping:3} (token independence) is trivially satisfied, since there are no token transfers among the contracts.
Instead, the contract $\contract{C0} \in \CmvCi$ is not sender-agnostic, thus violating condition~\eqref{th:lmev:contract-stripping:1}.
Indeed, MEV preservation does not hold, since:
\[
\lmev{\CmvD \cap \deps{\CmvC}}{\sysS}{\CmvC} 
= 
\lmev{\setenum{\contract{C0}}}{\sysS}{\setenum{\contract{C0}}} 
=
0
\]
This is because the adversary is restricted to calling $\contract{C0}$, but the transaction would revert since the $\txcode{require}$ condition in $\contract{C0}$ is violated.
\hfill\qedex
\end{example}

\begin{example}
\label{ex:lmev:contract-stripping:inclusion}
To illustrate the need of the assumption \eqref{th:lmev:contract-stripping:2} in \Cref{th:lmev:contract-stripping}, consider the contracts: 

\begin{lstlisting}[
,language=txscript
,basicstyle=\fontseries{m}\normalsize\ttfamily\lst@ifdisplaystyle\footnotesize\fi,
,morekeywords={f,g,f0,f1,g0,g1},classoffset=4
,morekeywords={a,A,M},keywordstyle=\pmvColor
,classoffset=5,morekeywords={t,T,T0,T1,T2,ETH}
,keywordstyle=\tokColor,classoffset=6
,morekeywords={C0,C1,C2,C3},keywordstyle=\cmvColor
%,caption={Illustration of~\Cref{th:lmev:contract-stripping}}
%,label={lst:cex:lmev:contract-stripping}
]
contract C0 { f() { transfer(M,1:T) } }
contract C1 { f() { C0.f(); } }
\end{lstlisting}

\noindent
Let 
$\Adv = \setenum{\pmvM}$, 
$\CmvC = \setenum{\contract{C0}}$, 
$\CmvD= \setenum{\contract{C1}}$, and let:
\[
\sysS = 
\walpmv{\pmvM}{0:\tokT} \mid
\walpmv{\contract{C0}}{\waltok{1}{\tokT}} \mid 
\walpmv{\contract{C1}}{\waltok{0}{\tokT}}
\]
Let $\TxTS \in \mall{\CmvD}{\pmvM}$ be the following sequence of transactions:
\[
\TxTS = \pmvM:\contract{C1}.\txcode{f()}
\]
By executing $\TxTS$ in $\sysS$, we have that:
\begin{align*}
    \sysS
    & \xrightarrow{\pmvM:\contract{C1}.\txcode{f()}}
    \walpmv{\pmvM}{1:\tokT} \mid
    \walpmv{\contract{C0}}{\waltok{0}{\tokT}} \mid 
    \walpmv{\contract{C1}}{\waltok{0}{\tokT}}
\end{align*}
Since there are no tokens left in $\CmvC$, $\TxTS$ clearly maximises the loss of $\CmvC$, hence:
\[
\lmev{\CmvD}{\sysS}{\CmvC} = 1 \cdot \price{\tokT}
\]
We have that $\CmvCi = \deps{\setenum{\contract{C0}}} \cap \deps{\setenum{\contract{C1}} \setminus \deps{\setenum{\contract{C0}}}} = \setenum{\contract{C0}} \not\subseteq \CmvD$, thus violating condition \eqref{th:lmev:contract-stripping:2}. 
We have that:
\[
\lmev{\CmvD \cap \deps{\CmvC}}{\sysS}{\CmvC} 
= 
\lmev{\setenum{\contract{C1}} \cap \setenum{\contract{C0}}}{\sysS}{\setenum{\contract{C0}}} 
= 
\lmev{\emptyset}{\sysS}{\setenum{\contract{C0}}} 
=
0
\]
Therefore, MEV preservation does not hold.
\hfill\qedex
\end{example}

\proofof{th:qnonint:preserving-interference}
We show the following two equalities, which imply the thesis:
\begin{align}
    \label{eq:qnonint:preserving-interference:1}
    \lmev{}{\sysS\mid\cstD}{\cmvOfcst{\cstD}} &= \lmev{}{\sysS \mid \cstC[\Adv]\mid\cstD}{\cmvOfcst{\cstD}}
    \\
    \label{eq:qnonint:preserving-interference:2}
    \lmev{\cmvOfcst{\cstD}}{\sysS\mid\cstD}{\cmvOfcst{\cstD}} &= \lmev{\cmvOfcst{\cstD}}{\sysS \mid \cstC[\Adv]\mid\cstD}{\cmvOfcst{\cstD}}
\end{align}
Observe that~\eqref{eq:qnonint:preserving-interference:2} follows directly from~\Cref{lem:lmev:prepending-state},
since $\sysS \mid \cstD \preceq \sysS \mid \cstC[\Adv] \mid \cstD$
and $\cmvOfcst{\cstD} \subseteq \cmvOfcst{(\sysS \mid \cstD)} = \cmvOfcst{\sysS} \cup \cmvOfcst{\cstD}$.
Note instead that~\eqref{eq:qnonint:preserving-interference:1} does not follow from~\Cref{lem:lmev:prepending-state}, since 
to equate
$\lmev{\cmvOfcst{(\sysS \mid \cstC[\Adv]\mid\cstD)}}{\sysS\mid\cstD}{\cmvOfcst{\cstD}}$ and $\lmev{\cmvOfcst{(\sysS \mid \cstC[\Adv]\mid\cstD)}}{\sysS \mid \cstC[\Adv]\mid\cstD}{\cmvOfcst{\cstD}}$,
the~\namecref{lem:lmev:prepending-state} would require 
$\cmvOfcst{(\sysS \mid \cstC[\Adv]\mid\cstD)} \subseteq \cmvOfcst{(\sysS \mid \cstD)}$, which is false.

% Indeed, widening the contract state does not increase the MEV extractable from $\cmvOfcst{\cstD}$ since the contracts that $\Adv$ is allowed to target is fixed.
%
\medskip\noindent
In order to prove~\eqref{eq:qnonint:preserving-interference:1}, 
we pass through two auxiliary results.
We start by proving the following equality:
\begin{align}
\label{eq:qnonint:preserving-interference:4}
\lmev
{\cmvOfcst{(\sysS \mid \cstC[\Adv] \mid \cstD)}}
{\sysS \mid \cstC[\Adv] \mid \cstD}{\cmvOfcst{\cstD}}
=
\lmev
{\cmvOfcst{(\sysS\mid\cstD)} \cap \deps{\cmvOfcst{\cstD}}}
{\sysS \mid \cstC[\Adv]\mid\cstD}{\cmvOfcst{\cstD}}
\end{align}
In order to apply~\Cref{th:lmev:contract-stripping}, let:
\[
\CmvC = \cmvOfcst{\cstD}
\qquad
\CmvD = \cmvOfcst{(\sysS \mid \cstC[\Adv] \mid \cstD)}
\]
and let:
\begin{align*}
\CmvCi 
& = \deps{\CmvC} \cap \deps{\CmvD \setminus \deps{\CmvC}}
\\
& = \deps{\cstD} \cap \deps{\cmvOfcst{(\sysS \mid \cstC[\Adv] \mid \cstD) \setminus \deps{\cstD}}}
\end{align*}
Note that the conditions of~\Cref{th:lmev:contract-stripping} are satisfied:
\begin{itemize}

\item[\eqref{th:lmev:contract-stripping:1}] $\CmvCi$ are sender-agnostic, since $\CmvCi \subseteq \deps{\cstD}$ and, by assumption of~\Cref{th:qnonint:preserving-interference}, the contracts in $\deps{\cstD}$ are sender-agnostic;

\item[\eqref{th:lmev:contract-stripping:2}] $\CmvCi \subseteq \CmvD$ holds since
$\CmvCi \subseteq \deps{\cstD} \subseteq \CmvD$;

\item[\eqref{th:lmev:contract-stripping:3}]
Since the state $\sysS \mid \cstD$ is well-formed by assumption, then $\deps{\cstD} \subseteq \cmvOfcst{(\sysS \mid \cstD)}$, and so we have that:
\begin{align*}
\deps{\CmvD} \cap \deps{\CmvC}
& = \deps{\sysS \mid \cstC[\Adv] \mid \cstD} \cap \deps{\cstD}
= \deps{\cstD}
\\
\deps{\CmvD} \setminus \deps{\CmvC}
& = \deps{\sysS \mid \cstC[\Adv] \mid \cstD} \setminus \deps{\cstD}
= \deps{\sysS \mid \cstC[\Adv]} \setminus \deps{\cstD}
\\
& \subseteq \cmvOfcst{(\sysS \mid \cstC[\Adv])} \setminus \deps{\cstD}
\end{align*}
Since $\sysS \mid \cstC[\Adv] \mid \cstD$ is well-formed and $\deps{\cstD}$ and $\cmvOfcst{(\sysS \mid \cstC[\Adv])} \setminus \deps{\cstD}$ are disjoint, then Condition~\ref{condition:qnonint:preserving-interference:2} of~\Cref{th:qnonint:preserving-interference} ensures that these sets are token independent.
% Condition~\ref{condition:qnonint:preserving-interference:2} 

\end{itemize}
Therefore, by~\Cref{th:lmev:contract-stripping} it follows that:
\begin{align*}
\lmev
{\cmvOfcst{(\sysS \mid \cstC[\Adv] \mid \cstD)}}
{\sysS \mid \cstC[\Adv] \mid \cstD}{\cmvOfcst{\cstD}}
=
\lmev
{\cmvOfcst{(\sysS \mid \cstC[\Adv] \mid \cstD)} \cap \deps{\cmvOfcst{\cstD}}}
{\sysS \mid \cstC[\Adv]\mid\cstD}{\cmvOfcst{\cstD}}
\end{align*}
To obtain~\eqref{eq:qnonint:preserving-interference:4}, just note that,
since $\deps{\cmvOfcst{\cstD}} \subseteq \cmvOfcst{(\sysS \mid \cstD)}$:
\[
\cmvOfcst{(\sysS \mid \cstC[\Adv] \mid \cstD)} \cap \deps{\cmvOfcst{\cstD}}
=
\cmvOfcst{(\sysS \mid \cstD)} \cap \deps{\cmvOfcst{\cstD}}
\]
which is equal to $\deps{\cmvOfcst{\cstD}}$.

\medskip\noindent
The second auxiliary result is the equality:
\begin{align}
    \label{eq:qnonint:preserving-interference:5}
    \lmev{\cmvOfcst{(\sysS\mid\cstD)}}{\sysS \mid \cstC[\Adv]\mid\cstD}{\cmvOfcst{\cstD}}
    =
    \lmev{\strip{\cmvOfcst{(\sysS \mid \cstD)}}{\cmvOfcst{\cstD}}}{\sysS \mid \cstC[\Adv]\mid\cstD}{\cmvOfcst{\cstD}}
\end{align}

\noindent
This time, in order to apply~\Cref{th:lmev:contract-stripping} we let: 
\[
\CmvC = \cmvOfcst{\cstD}
\qquad
\CmvD = \cmvOfcst{(\sysS\mid\cstD)}
\qquad
\CmvCi = \deps{\cstD} \cap \deps{\CmvD \setminus \deps{\CmvC}}
\]
\noindent
In order to apply~\Cref{th:lmev:contract-stripping}, let us first compute:
\begin{align*}
\CmvCi 
& = \deps{\CmvC} \cap \deps{\CmvD \setminus \deps{\CmvC}}
\\
& = \deps{\cstD} \cap \deps{\cmvOfcst{(\sysS \mid \cstD) \setminus \deps{\cstD}}}
\end{align*}
Again, note that the assumptions of~\Cref{th:lmev:contract-stripping} are satisfied:
\begin{itemize}

\item[\eqref{th:lmev:contract-stripping:1}] $\CmvCi$ are sender-agnostic, since $\CmvCi \subseteq \deps{\cstD}$ and assumption~\ref{condition:qnonint:preserving-interference:1};

\item[\eqref{th:lmev:contract-stripping:2}] $\CmvCi \subseteq \CmvD$ holds since $\CmvCi \subseteq \deps{\cstD} \subseteq \CmvD$;

\item[\eqref{th:lmev:contract-stripping:3}]
Since the state $\sysS \mid \cstD$ is well-formed by assumption, then $\deps{\cstD} \subseteq \cmvOfcst{(\sysS \mid \cstD)}$, and so we have that:
\begin{align*}
\deps{\CmvD} \cap \deps{\CmvC}
& = \deps{\sysS \mid \cstD} \cap \deps{\cstD}
= \deps{\cstD}
\\
\deps{\CmvD} \setminus \deps{\CmvC}
& = \deps{\sysS \mid \cstD} \setminus \deps{\cstD}
= \deps{\sysS} \setminus \deps{\cstD}
\end{align*}
Condition~\ref{condition:qnonint:preserving-interference:2} ensures that 
these sets are token independent.

% \color{red}{
% \medskip
% REMOVE RED PART we have to prove that
% \[
% \CmvD \cap \deps{\CmvC} = \deps{\cmvOfcst{\cstD}}
% \qquad
% \CmvD \setminus \deps{\CmvC} = \cmvOfcst{\sysS} \setminus \deps{\cmvOfcst{\cstD}} 
% \]
% are token independent in $\sysS \mid \cstC[\Adv] \mid \cstD$.
% This follows from condition~\ref{condition:qnonint:preserving-interference:2}, which gives the token independence of $\strip{\cmvOfcst{(\sysS \mid \cstD)}}{\cmvOfcst{\cstD}}$ and $\cmvOfcst{(\sysS \mid \cstC[\Adv])} \setminus \deps{\cstD}$ in $\sysS \mid \cstC[\Adv] \mid \cstD$, which in turns gives the token independence of $\strip{\cmvOfcst{(\sysS \mid \cstD)}}{\cmvOfcst{\cstD}}$ and $\cmvOfcst{\sysS} \setminus \deps{\cstD}$ in $\sysS \mid \cstD$ (to see this, substitute $\cstC[\Adv] = \emptyset$).}

\end{itemize}

\noindent
Therefore, \Cref{th:lmev:contract-stripping} gives the equality~\eqref{eq:qnonint:preserving-interference:4}.

\medskip\noindent
Now, by putting together \eqref{eq:qnonint:preserving-interference:4}
and \eqref{eq:qnonint:preserving-interference:5}, we obtain:
\begin{equation}
    \label{eq:qnonint:preserving-interference:6}
    \lmev{\cmvOfcst{(\sysS \mid \cstC[\Adv] \mid \cstD)}}{\sysS \mid \cstC[\Adv]\mid\cstD}{\cmvOfcst{\cstD}}
    =
    \lmev{\cmvOfcst{(\sysS \mid \cstD)}}{\sysS \mid \cstC[\Adv]\mid\cstD}{\cmvOfcst{\cstD}}
\end{equation}

\medskip\noindent
Now we can prove~\eqref{eq:qnonint:preserving-interference:1} by observing the following chain of equalities:
\begin{center}
\begin{align*}
    &\lmev{\cmvOfcst{(\sysS\mid\cstD)}}{\sysS \mid \cstC[\Adv]\mid\cstD}{\cmvOfcst{\cstD}}
    \xlongequal{\text{by~\eqref{eq:qnonint:preserving-interference:6}}} \lmev{\cmvOfcst{(\sysS \mid \cstC[\Adv]\mid\cstD)}}{\sysS \mid \cstC[\Adv]\mid\cstD}{\cmvOfcst{\cstD}}
    % && \text{by~\eqref{eq:qnonint:preserving-interference:6}}
    \\
    & \hspace{15mm} \biggm\Vert \text{by \Cref{lem:lmev:prepending-state}} \hspace{34.5mm} \biggm\Vert &&
    \\
    & \lmev{\cmvOfcst{(\sysS\mid\cstD)}}{\sysS\mid\cstD}{\cmvOfcst{\cstD}}
    \hspace{23mm} \lmev{\cmvOfcst{(\sysS \mid \cstC[\Adv]\mid\cstD)}}{\sysS \mid \cstC[\Adv]\mid\cstD}{\cmvOfcst{\cstD}}
    \\
    & \hspace{15mm} 
    \biggm\Vert \text{by~\Cref{lem:lmev}\eqref{lem:lmev:garbage}} \hspace{30mm} \biggm\Vert \text{by~\Cref{lem:lmev}\eqref{lem:lmev:garbage}}
    % &&
    \\
    &\lmev{}{\sysS\mid\cstD}{\cmvOfcst{\cstD}}
    \hspace{32.5mm} \lmev{}{\sysS \mid \cstC[\Adv]\mid\cstD}{\cmvOfcst{\cstD}}
\end{align*}
\end{center}
Now, the thesis directly follows from~\Cref{eq:qnonint:preserving-interference:1,eq:qnonint:preserving-interference:2}.
% \begin{align}
%     \label{eq:qnonint:preserving-interference:4}
%     \lmev{}{\sysS \mid \cstC[\Adv]\mid\cstD}{\cmvOfcst{\cstD}}
%     =
%     \lmev{\deps{\cstD}}{\sysS \mid \cstC[\Adv]\mid\cstD}{\cmvOfcst{\cstD}}
% \end{align}
\qed
\section{Proofs: use cases}
\label{sec:proofs:use-cases}

\subsection*{AMM/Bet (\Cref{ex:amm-bet})}

Consider the starting state:
\begin{align*}    
\sysS & = \walu{\pmvM}{m}{\ETH} \mid
\walpmv{\contract{AMM}}{\waltok{r_0}{\ETH},\waltok{r_1}{\tokT}} \mid
\code{block.num} = d-k \mid \cdots
\\
\cstD & = \walpmv{\contract{Bet}}{\waltok{b}{\ETH}, \code{owner}=\pmvA, \code{tok}=\tokT, \code{rate}=r, \code{deadline}=d}
\end{align*}
When $\pmvM$ is allowed to manipulate the $\contract{AMM}$, she can inflate the exchange rate of $\ETH$,
provided that she possesses sufficient funds.
Formally, if $\pmvM$ swaps $x:\ETH$ for $y:\tokT$,
then according to the criterion specified in $\contract{Bet}.\txcode{win}()$, the winner receives an amount $\lfloor 2bp \rfloor$ only if 
$\contract{AMM}.\txcode{getRate(\ETH)} = \nicefrac{r_0 + x}{r_1 - y} \geq p \cdot r$.
Assuming that $\pmvM$ enters the bet only when she can choose $x$ sufficiently high to satisfy this condition, and for $p \geq \nicefrac{1}{2}$ (since a smaller proportion makes the bet irrational for her), she fires the following sequence of transactions: where, in the $\txcode{swap}$ transaction, $x = m-b \geq 0$ is the number of $\ETH$ units sent to the $\contract{AMM}$, $y = \left\lfloor \nicefrac{x r_1}{r_0 + x} \right\rfloor$ is the number of $\tokT$ units received, and the value that $\pmvM$ bets on is $p = \nicefrac{r_0 + x}{r(r_1 - y)}$:
\begin{align*}
    \sysS \mid \cstD
    & \xrightarrow{\pmvM:\contract{Bet}.\txcode{bet}(\pmvM\ \code{pays}\ \waltok{b}{\ETH},p)}
    &&
    % \walu{\pmvM}{x}{\ETH} \mid
    \walpmv{\contract{AMM}}{\waltok{r_0}{\ETH},\waltok{r_1}{\tokT}} \mid
    \walpmv{\contract{Bet}}{\waltok{2b}{\ETH},\code{potShare}=p,\cdots} \mid \cdots
    \\
    & \xrightarrow{\pmvM:\contract{AMM}.\txcode{swap}(\pmvM\ \code{pays}\ \waltok{x}{\ETH}, 0)}
    &&
    % \walu{\pmvM}{y}{\tokT} \mid
    \walpmv{\contract{AMM}}{\waltok{r_0 + x}{\ETH},\waltok{r_1 - y}{\tokT}} \mid \walpmv{\contract{Bet}}{\waltok{2b}{\ETH},\cdots} \mid \cdots
    \\
    & \xrightarrow{\pmvM:\contract{Bet}.\txcode{win}()}
    &&
    % \walpmv{\pmvM}{\waltok{2bp}{\ETH},\waltok{y}{\tokT}} \mid
    \walpmv{\contract{AMM}}{\waltok{r_0 + x}{\ETH},\waltok{r_1 - y}{\tokT}} \mid \walpmv{\contract{Bet}}{\waltok{2b- \lfloor2bp\rfloor}{\ETH},\cdots} \mid \cdots
    \\
    & \xrightarrow{\pmvM:\contract{AMM}.\txcode{swap}(\pmvM\ \code{pays}\ \waltok{y}{\tokT}, 0)}
    &&
    % \walu{\pmvM}{2bp + x}{\ETH} \mid
    \walpmv{\contract{AMM}}{\waltok{r_0}{\ETH},\waltok{r_1}{\tokT}} \mid
    \walpmv{\contract{Bet}}{\waltok{2b- \lfloor2bp\rfloor}{\ETH},\cdots} \mid \cdots
\end{align*}
By~\Cref{eq:lmev:unrestricted} we have:
\begin{align*}
    \lmev{}{\sysS \mid \cstD}{\setenum{\contract{Bet}}}
    &= b - (2b- \lfloor2bp\rfloor) \;
    = \lfloor2bp\rfloor - b \;
    \leq 2bp - b \;
    = \frac{2b(r_0 + x)}{r(r_1 - y)} - b \;
    \\
    &= \frac{2b(r_0 + x)}{r \left( r_1 - \left\lfloor \frac{x r_1}{r_0 + x} \right\rfloor \right)} - b \;
    \leq \frac{2b(r_0 + x)}{r \left( r_1 - \frac{x r_1}{r_0 + x} \right)} - b \;
    \\\\
    &= \frac{2b(r_0 + x)^2}{r r_0 r_1} - b
    = \left( \frac{2 (r_0 + m - b)^2}{r r_0 r_1} - 1 \right) b
\end{align*}
Whereas, if $\pmvM$ was restricted to interact with $\contract{Bet}$ only, there are two cases: if $\contract{AMM}.\code{getRate} = \nicefrac{r_0}{r_1} \geq p \cdot r$,
then $\pmvM$ wins the bet. Otherwise, she loses (and, therefore, $\contract{Bet}$ does not suffer an economic loss).
Even in this case, $\pmvM$ enters the bet only for $p \geq \nicefrac{1}{2}$.
Therefore~\Cref{eq:lmev} gives us:
\begin{align*}
    \lmev{\setenum{\contract{Bet}}}{\sysS \mid \cstD}{\setenum{\contract{Bet}}}
    &= \begin{cases}
        b - \left(2b- \left\lfloor\frac{2br_0}{rr_1}\right\rfloor \right) & \text{if}\  \frac{r_0}{r r_1} \geq \nicefrac{1}{2} \\
        0 & \text{otherwise}
    \end{cases}
    \\
    &= \begin{cases}
        \left\lfloor\frac{2br_0}{rr_1}\right\rfloor - b & \text{if}\  \frac{r_0}{r r_1} \geq \nicefrac{1}{2} \\
        0 & \text{otherwise}
    \end{cases}
    \\
    &> \begin{cases}
        \frac{2br_0}{rr_1} -b -1 & \text{if}\  \frac{r_0}{r r_1} \geq \nicefrac{1}{2} \\
        0 & \text{otherwise}
    \end{cases}
\end{align*}
Hence $\lmev{}{}{}$ interference is estimated through \Cref{def:qnonint} as follows:
\begin{align*}
    \qnonint{\sysS}{\cstD}
    &< \begin{cases}
        1 - \frac{\frac{2br_0}{rr_1} -b -1}{\left( \frac{2 (r_0 + m - b)^2}{r r_0 r_1} - 1 \right) b}
        & \text{if}\ \frac{r_0}{r r_1} \geq \nicefrac{1}{2} \\
        1 & \text{otherwise}
        \end{cases}
    \\
    &=  \begin{cases}
        1 - \frac{(2br_0 - brr_1 - rr_1) r_0}{b(2(r_0+m-b)^2 - rr_0r_1)}
        & \text{if}\ \frac{r_0}{r r_1} \geq \nicefrac{1}{2} \\
        1 & \text{otherwise}
        \end{cases}
    \\
    &=  \begin{cases}
        1 - \frac{2br_0^2 - rr_0r_1(b+1)}{2b(r_0+m-b)^2 - brr_0r_1}
        & \text{if}\ \frac{r_0}{r r_1} \geq \nicefrac{1}{2} \\
        1 & \text{otherwise}
        \end{cases}
        \tag*{\qedex}
\end{align*}

\subsection*{AMM/Lending Pool (\Cref{ex:amm-LP})}

For simplicity, we make the following assumptions:
\begin{inlinelist}
\item $\price{\ETH} = 1 = \price{\tokT}$,
\item the $\contract{AMM}$ is balanced, 
\item $\pmvM$ has not deposited or borrowed tokens from the $\contract{LP}$ yet.
\item the $\contract{LP}$ has sufficient reserves of $\tokT$ to satisfy any borrow request. % \ie, $b \geq t(x)$.
\end{inlinelist}
Note also that our simplified $\contract{LP}$ contract only offers two functions, $\txcode{deposit}$ and $\txcode{borrow}$.
Calling $\txcode{deposit}$ does not extract tokens from the $\contract{LP}$, so the only action through which $\pmvM$ could cause a loss to the $\contract{LP}$ is $\txcode{borrow}$.

Consider the following blockchain state:
\begin{align*}    
\sysS = \walu{\pmvM}{n}{\ETH} \mid
\walpmv{\contract{AMM}}{\waltok{r}{\ETH},\waltok{r}{\tokT}}
&&\cstD = \walpmv{\contract{LP}}{\waltok{a}{\ETH}, \waltok{b}{\tokT}, \code{Cmin} = C_{\it min}, \cdots}
\end{align*}

% To simplify the computations in this, we assume that $\contract{LP}$ permits users to trade real-valued amounts of tokens (\ie \mbox{$\wmvA \in \TokU \rightarrow \mathbb{R}$}).

We start by estimating the unrestricted local MEV, \ie $\lmev{}{\sysS \mid \cstD}{\setenum{\contract{LP}}}$.
When $\pmvM$ can interact with the $\contract{AMM}$, she can maximize the loss caused to $\contract{LP}$ by maximizing her loan amount, or in other words, by inflating her collateralization ratio.
There is only one way to do so: by depositing a portion of her $\ETH$ to the $\contract{LP}$ and by inflating the exchange rate of $\ETH$ provided by the $\contract{AMM}$.
To this purpose, $\pmvM$ partitions its funds as follows:
\begin{itemize}
    \item $x:\ETH$ to perform a swap in the $\contract{AMM}$ in exchange for $y:\tokT$, where $y = \nicefrac{x r}{r + x}$
    \item $(n-x):\ETH$ to deposit in the $\contract{LP}$
\end{itemize}
We denote by $t(x)$ the number of units of token $\tokT$ that $\pmvM$ can borrow from the $\contract{LP}$ as a function of $x$. 

We first note that in order to satisfy the $\code{require}$ constraint within $\txcode{borrow}$ function of $\contract{LP}$, $\pmvM$ must be over-collateralized in the new $\contract{LP}$ state.
Recall that the collateralization of a user is given by the ratio between the value of her minted tokens and that of her debts. 
Regarding $\pmvM$, the value $v_{\it minted}$ of her minted tokens and the value $v_{\it debt}$ of her debts in the new state are given by: 
\begin{align*}
v_{\it minted} & = (n-x) \cdot \frac{\contract{AMM}.\txcode{getRate}(\ETH)}{\contract{AMM}.\txcode{getRate}(\tokT)}
\; = \;
(n-x) \cdot \frac{r - y}{r + x}
\\
v_{\it debt} & = t(x) \cdot \frac{\contract{AMM}.\txcode{getRate}(\tokT)}{\contract{AMM}.\txcode{getRate}(\ETH)}
\; = \;
t(x) \cdot \frac{r + x}{r - y}
\end{align*}
Therefore, $\pmvM$ is over-collateralized, and so her call to $\txcode{borrow}$ does not revert, if:
\[
% \txcode{collateral(\pmvM)} =
\frac{v_{\it minted}}{v_{\it debt}}
\; = \;
\frac{(n-x) (r + x)^2}{t(x) (r - y)^2} 
\; \geq \;
C_{min}
\]
This gives us the maximum value of $t(x)$ that $\pmvM$ can choose, which is:
\begin{align*}
    t(x) = \frac{(n-x) (r + x)^2}{C_{min} (r - y)^2}
\end{align*}
To find the value of $x$ that maximizes $t(x)$, we study the function $t(x)$ that gives the loan amount as a function of the deposited amount $x$, subject to the constraint $0 \leq x \leq n$.
Since we working with real-valued amounts, we have that $t(x)$ is continuous.
Thus, we compute its derivative \wrt $x$ and set it to $0$:
\begin{align*}
    \frac{d t(x)}{dx} \;
    = \; \frac{d}{dx} \left( \frac{(n-x) (r + x)^4}{C_{min} r^4} \right) \;
    = \; \frac{4(n-x)(r+x)^3 - (r+x)^4}{C_{min} r^4} = \; 0
\end{align*}
Since $r+x > 0$, we can simplify the above as:
\begin{align*}
    & 4(n-x) = r+x
    && \text{if $0 \leq x \leq n$}
\end{align*}
Therefore, the $x$ that maximizes $t(x)$ is given by:
\begin{align*}
    x & = \begin{cases} 
    \frac{4n - r}{5} & \text{if $4 n \geq r$}
    \\
    0 & \text{otherwise}
    \end{cases}
\end{align*}
% Therefore, $x = \frac{4n-r}{5}$ maximizes $t(x)$ when $4n \geq r$.
% Otherwise, the value of $x$ that maximizes $t(x)$ is $x = 0$.
In other words, when $4n < r$, $\pmvM$ does not need to interact with the $\contract{AMM}$ to maximize her borrowing capacity.

\medskip\noindent
We can check that $x = \frac{4n-r}{5}$ maximizes $t(x)$ by performing the double derivative test.
We compute the double derivative of $t(x)$ w.r.t $x$, plugging in $x = \frac{4n-r}{5}$, and check if it is $< 0$.
Accordingly:
\begin{align*}
    \frac{d^2 t(x)}{dx^2} \;
    &= \; \frac{d}{dx} \left( \frac{4(n-x)(r+x)^3 - (r+x)^4}{C_{min} \cdot r^4}\right)
    \\\\
    &= \; \frac{12(n-x)(r+x)^2 - 4(r+x^3) - 4(r+x)^3}{C_{min} \cdot r^4}
    \\\\
    &= \; \frac{12(n-x)(r+x)^2 - 8(r+x)^3}{C_{min} \cdot r^4}
\end{align*}
Substituting $4(n-x) = r+x$ we get:
\begin{align*}
    &\frac{d^2 t(x)}{dx^2} \;
    = \; \frac{3(r+x)^3 - 8(r+x)^3}{C_{min} \cdot r^4} \;
    = \; - \frac{5(r+x)^3}{C_{min} \cdot  r^4}
    < 0
\end{align*}
As a result, $\pmvM$ fires the following sequence of transactions with a loan amount $t = \nicefrac{(n-x) (r + x)^2}{C_{\it min} (r - y)^2}$ and the amount received on swap $y = \nicefrac{x r}{r + x}$:
\begin{align*}
    \sysS \mid \cstD
    & \xrightarrow{\pmvM:\contract{LP}.\txcode{deposit}(\pmvM\ \code{pays}\ \waltok{(n-x)}{\ETH})}
    &&
    % \walu{\pmvM}{x}{\ETH} \mid
    \walpmv{\contract{AMM}}{\waltok{r}{\ETH},\waltok{r}{\tokT}}
    \mid
    \walpmv{\contract{LP}}{\waltok{a + n - x}{\ETH}, \waltok{b}{\tokT}, \cdots} \mid \cdots
    \\
    & \xrightarrow{\pmvM:\contract{AMM}.\txcode{swap}(\pmvM\ \code{pays}\ \waltok{x}{\ETH}, 0)}
    &&
    % \walu{\pmvM}{y}{\tokT} \mid
    \walpmv{\contract{AMM}}{\waltok{r + x}{\ETH},\waltok{r - y}{\tokT}} \mid
    \walpmv{\contract{LP}}{\waltok{a + n - x}{\ETH},\waltok{b}{\tokT}, \cdots} \mid \cdots
    \\
    & \xrightarrow{\pmvM:\contract{LP}.\txcode{borrow}(t,\tokT)}
    &&
    % \walu{\pmvM}{y + t}{\tokT} \mid
    \walpmv{\contract{AMM}}{\waltok{r + x}{\ETH},\waltok{r - y}{\tokT}} \mid
    \walpmv{\contract{LP}}{\waltok{a + n - x}{\ETH},\waltok{b - t}{\tokT}, \cdots} \mid \cdots
    \\
    & \xrightarrow{\pmvM:\contract{AMM}.\txcode{swap}(\pmvM\ \code{pays}\ \waltok{y}{\tokT}, 0)}
    &&
    % \walpmv{\pmvM}{\waltok{x}{\ETH},\waltok{t}{\tokT}} \mid
    \walpmv{\contract{AMM}}{\waltok{r}{\ETH},\waltok{r}{\tokT}}
    \mid
    \walpmv{\contract{LP}}{\waltok{a + n - x}{\ETH},\waltok{b - t}{\tokT}, \cdots} \mid \cdots
\end{align*}
By~\Cref{eq:lmev:unrestricted} we get:
\begin{align*}
    \lmev{}{\sysS \mid \cstD}{\setenum{\contract{LP}}} 
    & = t + x - n \;
    = \frac{(n - x) (r + x)^2}{C_{min} (r - y)^2} + x - n \;
    \\\\
    & = (n - x) \left( \frac{(r + x)^2}{C_{\it min} (r - \frac{xr}{r+x})^2} - 1 \right) \;
    \\\\
    & = (n - x) \left( \frac{(r + x)^4}{r^4C_{\it min}} - 1 \right) \;
    \\\\
    & = \begin{cases}
        \left(n - \frac{4n - r}{5}\right) \left( \frac{\left(r + \frac{4n - r}{5}\right)^4}{r^4C_{\it min}} - 1 \right) \;
        & \text{if}\ 4n \geq r \\
        n \left( \frac{1}{C_{min}} - 1 \right) & \text{otherwise}
        \end{cases}
    \\\\
    & = \begin{cases}
        \left( \frac{n+r}{5} \right) \left( \frac{\left(\frac{4(n+r)}{5}\right)^4}{r^4 C_{min}} - 1 \right) \;
        & \text{if}\ 4n \geq r \\
        n \left( \frac{1}{C_{min}} - 1 \right) & \text{otherwise}
        \end{cases}
    \\\\
    & = \begin{cases}
        \left( \frac{n+r}{5} \right) \left( \frac{1}{C_{min}} \left( \frac{4(n+r)}{5r}\right)^4 - 1 \right) \;
        & \text{if}\ 4n \geq r \\
        n \left( \frac{1}{C_{min}} - 1 \right) & \text{otherwise}
        \end{cases}
\end{align*}

We note two key aspects of the transaction sequence fired by $\pmvM$.
Firstly, the ordering of $\txcode{deposit}$ and the (initial) $\txcode{swap}$ transactions is irrelevant.
Hence, they can be interchanged without causing a difference to the loss caused to $\contract{LP}$.
Secondly, firing the (final) $\txcode{swap}$, \ie de-manipulating the $\contract{AMM}$ only affects the wealth of $\pmvM$ and not the $\contract{LP}$.
Hence, it does not affect the MEV extractable from $\contract{LP}$.
Nevertheless, we include it in the transaction sequence to reflect the attack execution employed in practice.

We now calculate the restricted local MEV, \ie $\lmev{\setenum{\contract{LP}}}{\sysS \mid \cstD}{\setenum{\contract{LP}}}$.
In this case, the only way $\pmvM$ can maximize her borrowing capacity is by depositing her total available capital to the $\contract{LP}$.
Hence, $\pmvM$ deposits $n:\ETH$. 
The collateralization of $\pmvM$ after a call to $\txcode{borrow}$ 
for $t'$ units of $\tokT$ is given by:
\begin{align*}
\frac{v_{\it minted}}{v_{\it debt}}
& = 
\frac{n \cdot \nicefrac{\contract{AMM}.\txcode{getRate}(\ETH)}{\contract{AMM}.\txcode{getRate}(\tokT)}}
{t' \cdot \nicefrac{\contract{AMM}.\txcode{getRate}(\tokT)}{\contract{AMM}.\txcode{getRate}(\ETH)}}
\; = \;
\frac{n \cdot \nicefrac{r}{r}}
{t' \cdot \nicefrac{r}{r}}
\; = \;
\frac{n}{t'(x)}
\end{align*}
Thus, the call to $\txcode{borrow}$ does not revert iff:
\[
    \frac{n}{t'} \geq C_{min}
\]
From this, we obtain that the maximum amount that $\pmvM$ can borrow is given by:
\[
    t' = \frac{n}{C_{min}}
\]
By~\Cref{eq:lmev}, we have that:
\begin{align*}
    \lmev{\setenum{\contract{LP}}}{\sysS \mid \cstD}{\setenum{\contract{LP}}}
    = \; t' - n \;
    = \; \frac{n}{C_{min}} - n \;
    = \; n \left( \frac{1}{C_{min}} - 1 \right)
\end{align*} 
To conclude, we estimate $\lmev{}{}{}$ interference through \Cref{def:qnonint} as follows:
\begin{align*}
    \qnonint{\sysS}{\cstD}
    &= \begin{cases}
        1 \; - \; \frac{n \left( \frac{1}{C_{min}} - 1 \right)}{\left( \frac{n+r}{5} \right) \left( \frac{1}{C_{min}} \left( \frac{4(n+r)}{5r}\right)^4 - 1 \right)} \;
        & \text{if}\ 4n \geq r
        \\\\
        1 \; - \; \frac{n \left( \frac{1}{C_{min}} - 1 \right)}{n \left( \frac{1}{C_{min}} - 1 \right)} \; & \text{otherwise}
        \end{cases}
    \\\\
    &= \begin{cases}
        1 \; - \; \frac{n(1-C_{min})}{C_{min}} \cdot \frac{5}{n+r} \cdot \frac{(5r)^4C_{min}}{(4(n+r))^4 - (5r)^4C_{min}} \;
        & \text{if}\ 4n \geq r \\
        0 \; & \text{otherwise}
        \end{cases}
    \\\\
    &= \begin{cases}
        1 \; - \; \frac{5^5 r^4 n(1-C_{min})}{(n+r)\left( 4^4(n+r)^4 - (5r)^4C_{min}\right)} \;
        & \text{if}\ 4n \geq r \\
        0 \; & \text{otherwise}
        \end{cases}
    \tag*{\qedex}
\end{align*}

% \bartnote{an under-approximation of MEV interference:}
% \begin{align*}
% & \lmev{}{\sysS \mid \cstD}{\setenum{\contract{LP}}} 
% \geq \encircle{a}
% \\
% & \lmev{\setenum{\contract{LP}}}{\sysS \mid \cstD}{\setenum{\contract{LP}}}
% \leq \encircle{b}
% \\
% & \qnonint{\sysS}{\cstD}
% \geq 
% 1 - \frac{\encircle{b}}{\encircle{a}}
% \end{align*}
}
{}

\end{document}